\documentclass{article}

\usepackage[numbers, compress]{natbib}

\usepackage[final]{neurips_2022}

\usepackage{wrapfig}
\usepackage[utf8]{inputenc} 
\usepackage[T1]{fontenc}    
\usepackage{hyperref}       
\usepackage{url}            
\usepackage{booktabs}       
\usepackage{amsfonts}       
\usepackage{nicefrac}       
\usepackage{microtype}      
\usepackage{xcolor}         
\usepackage{xspace}
\usepackage{bm}
\usepackage{mathtools}
\usepackage{caption}
\usepackage{subcaption}
\usepackage{enumitem}

\def\eg{{\em e.g.,}\xspace}
\def\ie{{\em i.e.,}\xspace}

\def\etc{{\em etc.}\xspace}
\def\wrt{{\em w.r.t.}\xspace}

\usepackage{amsmath}
\usepackage{amssymb}
\usepackage{mathtools}
\usepackage{amsthm}

\theoremstyle{plain}
\newtheorem{theorem}{Theorem}[section]

\newtheorem{lemma}[theorem]{Lemma}

\theoremstyle{definition}

\theoremstyle{remark}

\definecolor{deepred}{rgb}{0.631,0.102,0.102}

\def\cater{CATER\xspace}

\newcommand{\minus}{\scalebox{0.75}[1.0]{$-$}}

\def\figref#1{Figure~\ref{#1}}
\def\Figref#1{Figure~\ref{#1}}

\def\tabref#1{Table~\ref{#1}}
\def\Tabref#1{Table~\ref{#1}}

\def\Secref#1{Section~\ref{#1}}

\def\eqref#1{(\ref{#1})}
\def\Eqref#1{Equation~\ref{#1}}

\def\1{\bm{1}}

\def\vzero{{\bm{0}}}
\def\vone{{\bm{1}}}

\def\vc{{\bm{c}}}

\def\vy{{\bm{y}}}

\def\mI{{\bm{I}}}

\def\mM{{\bm{M}}}

\def\mP{{\bm{P}}}
\def\mQ{{\bm{Q}}}

\def\mW{{\bm{W}}}
\def\mX{{\bm{X}}}

\usepackage{pifont}
\definecolor{myred}{RGB}{215,48,39}
\definecolor{mygreen}{RGB}{26,152,80}

\title{\cater: Intellectual Property Protection on Text Generation APIs via Conditional Watermarks}

\author{
Xuanli He\thanks{Equal contributions. Most of the work was finished when X.H was at Monash Unversity.} \\ University College London \\ zodiac.he@gmail.com \And Qiongkai Xu \footnotemark[1] \\ University of Melbourne \\ qiongkai.xu@unimelb.edu.au \And Yi Zeng  \\ Virginia Tech \\yizeng@vt.edu\AND Lingjuan Lyu\thanks{Corresponding author} \\ Sony AI \\ Lingjuan.Lv@sony.com \And Fangzhao Wu \\ Microsoft Research Asia \\ fangzwu@microsoft.com \And Jiwei Li \\ Shannon.AI, Zhejiang University \\ jiwei\_li@shannonai.com \And Ruoxi Jia \\ Virginia Tech \\ ruoxijia@vt.edu \\
}

\begin{document}

\maketitle

\begin{abstract}
Previous works have validated that text generation APIs can be stolen through imitation attacks, causing IP violations. In order to protect the IP of text generation APIs, a recent work has introduced a watermarking algorithm and utilized the null-hypothesis test as a post-hoc ownership verification on the imitation models. However, we find that it is possible to detect those watermarks via sufficient statistics of the frequencies of candidate watermarking words. To address this drawback, in this paper, we propose a novel Conditional wATERmarking framework (\cater) for protecting the IP of text generation APIs. An optimization method is proposed to decide the watermarking rules that can minimize the distortion of overall word distributions while maximizing the change of conditional word selections. Theoretically, we prove that it is infeasible for even the savviest attacker (they know how CATER works) to reveal the used watermarks from a large pool of potential word pairs based on statistical inspection. Empirically, we observe that high-order conditions lead to an exponential growth of suspicious (unused) watermarks, making our crafted watermarks more stealthy. In addition, \cater can effectively identify the IP infringement under architectural mismatch and cross-domain imitation attacks, with negligible impairments on the generation quality of victim APIs. We envision our work as a milestone for stealthily protecting the IP of text generation APIs.
\end{abstract}

\section{Introduction}
Nowadays, many technology corporations, such as Google, Amazon, Microsoft, have invested a plethora of workforce and computation to data collection and model training, in order to
deploy well-trained commercial models as pay-as-you-use services on their cloud platforms. Therefore, these corporations own the intellectual property (IP) of their trained models. Unfortunately, previous works have validated that the functionality of a victim API can be stolen through imitation attacks, which inquire the victim with carefully designed queries and train an imitation model based on the outputs of the target API. Such attacks cause severe IP violations of the target API and stifle the creativity and motivation of our research community~\cite{tramer2016stealing, wallace2020imitation, Krishna2020Thieves, he2021model}. 

In fact, imitation attacks work not only on laboratory models, but also on commercial APIs \citep{wallace2020imitation, xu2021beyond}, since the enormous commercial benefit allures competing companies or individual users to extract or steal these successful APIs. For instance, some leading companies in NLP business have been caught imitating their competitors' models~\citep{singhal2011google}. Beyond imitation attacks, the attacker could potentially surpass victims by conducting unsupervised domain adaptation and multi-victim ensemble~\citep{xu2021beyond}. 

\begin{figure}
    \centering
    \includegraphics[width=0.9\textwidth]{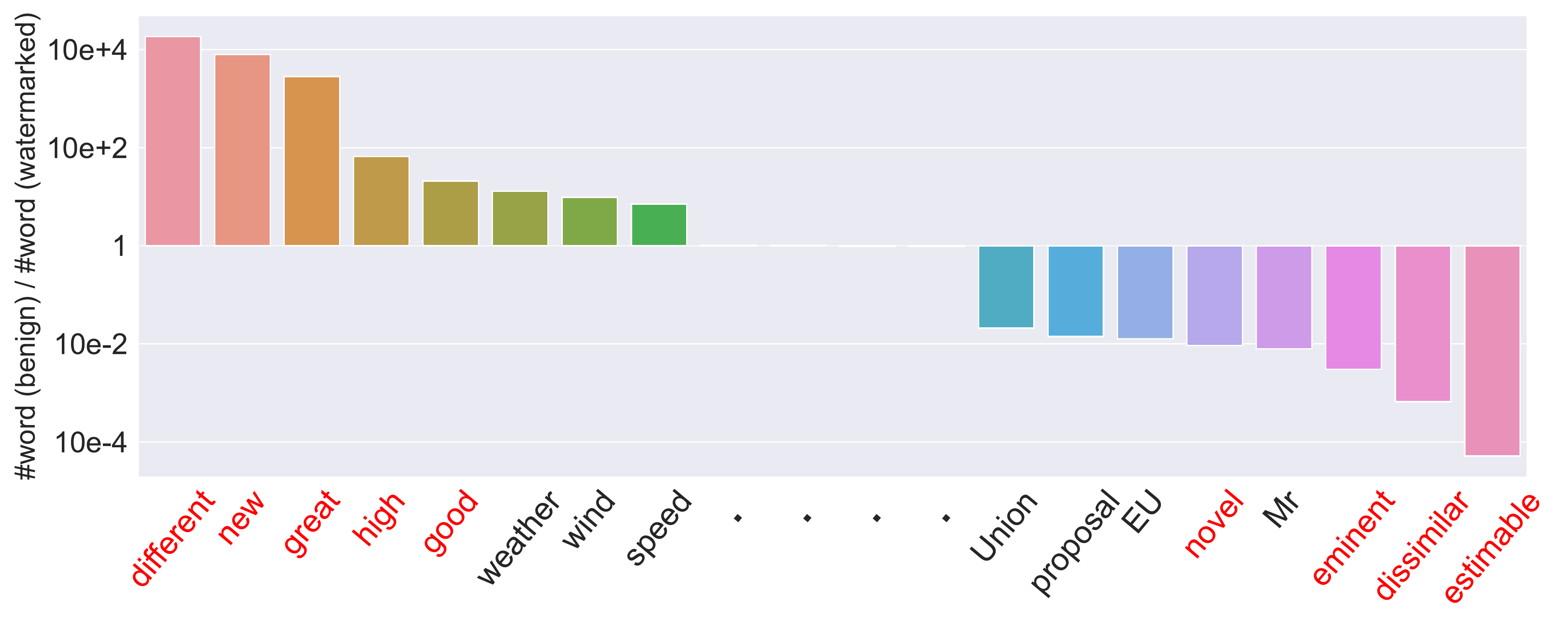}
    \caption{Ratio change of word frequency of top 100 words between benign and watermarked corpora used by~\cite{he2022protecting}, namely $P_b(w)/P_w(w)$. \textcolor{red}{Red} words are the selected watermarks. Although we only list 16 words having the most significant ratio change in the benign and watermarked corpora and omit the rest of them for better visualization, all watermarks are within the top 100 words.}
    \label{fig:word_dist}
    \vspace{-5mm}
\end{figure}

In order to protect the IP of victim models, He et al.~\cite{he2022protecting} first introduced a watermarking algorithm to text generation and utilized the null-hypothesis test as a post-hoc ownership verification on the imitation models. However, traditional watermarking methods generally distort the word distribution, which could be utilized by attackers to infer the watermarked words via sufficient statistics of the frequency change of candidate watermarking words. As an example shown in~\Figref{fig:word_dist}, the replaced words and their substitutions are those with most frequency decrease ratios and frequency increase ratios, respectively. To address this drawback, we are motivated to develop a more stealthy watermarking method to protect the IP of text generation APIs. The stealthiness of the new watermarks is achieved by incorporating high-order linguistic features as conditions that trigger corresponding watermarking rules.

Overall, our main contributions are as follows\footnote{Code and data are available at: \url{https://github.com/xlhex/cater_neurips.git}}:
\begin{itemize}[topsep=0pt, partopsep=0pt, leftmargin=10pt, parsep=0pt, itemsep=1.75pt]
    \item We propose a novel Conditional wATERmarking framework (\cater) for protecting text generation APIs. An optimization method is proposed to decide the watermarking rules that can \textit{i)} minimize the distortion of overall word distributions, while \textit{ii)} maximize the change of conditional word selections. 
    \item Theoretically, we prove that a small number of the used watermarks could be blended in and camouflaged by
    a large number of suspicious watermarks when attackers attempt to inverse the backend watermarking rules. 
    \item Empirically, we observe that high-order conditions lead to the exponential growth of suspicious (unused) watermarks, which encourages better stealthiness of the proposed defense with little hurt to the generation of victim models.
\end{itemize}
\section{Preliminary and Background}
\subsection{Imitation Attack} 
An imitation attack (a.k.a model extraction) aims to emulate the behavior of the victim model $\mathcal{V}$, such that the adversary can either sidestep the service charges or launch a competitive service~\citep{tramer2016stealing,Krishna2020Thieves,wallace2020imitation,he2021model,xu2021beyond}. Malicious users can achieve this goal through interaction with the victim model $\mathcal{V}$ without knowing its internals, such as the model architecture, hyperparameters, training data, \etc Adversaries first craft a set of queries $Q$ based on the documentation of a target model. Then $Q$ will be sent to $\mathcal{V}$ to obtain the corresponding predictions $Y$. Finally, an imitation model $\mathcal{S}$ can be attained by learning a function to map $Q$ to $Y$.


Most prior imitation attacks are limited to classification tasks~\citep{tramer2016stealing,orekondy2019knockoff,he2021model}. The imitation for text generation, a crucial task in natural language processing, has been under-developed until recently. Inspired by the efficacy of sequence-level knowledge distillation~\citep{kim2016sequence}, Wallace et al.~\cite{wallace2020imitation} and Xu et al.~\cite{xu2021beyond} propose mimicking the functionally of commercial text generation APIs. Similar to the standard imitation attack, adversaries can query $\mathcal{V}$ with $Q$. For generation tasks, $Y$ is a sequence of tokens $(y_1,...,y_L)$, where $L$ is the length of the sequence. According to their empirical studies, one can rival the performance of these APIs, which poses a severe threat to cloud platforms.

\subsection{Identification of IP Infringement}
Prior works have utilized watermarking avenues to achieve a post-hoc verification of the ownership
~\citep{uchida2017embedding,li2020protecting,lim2022protect}. However, this line of work assumes the model owners can watermark victim model $\mathcal{V}$ by altering its neurons before releasing $\mathcal{V}$ to end-users. This operation is not feasible for the imitation attack, as $\mathcal{V}$ cannot access the parameters of $\mathcal{S}$. The only thing under the control of $\mathcal{V}$ is the responses to adversaries. Hence, some recent works propose creating a backdoor to $\mathcal{S}$ during the interaction with attackers~\citep{Krishna2020Thieves,szyller2021dawn}. Specifically, $\mathcal{V}$ can select a small fraction of queries and answer them with incorrect predictions, in a similar way to the popular choice of watermarks in the computer vision domain (adopting some arbitrary features as the trigger to evaluate) \citep{guo2020fine,zeng2021rethinking}. 
Erroneous predictions are so abrupt that $\mathcal{S}$ will memorize these outliers~\citep{dai2019backdoor,kurita2020weight}. As such, $\mathcal{V}$ can utilize these watermarks as evidence of ownership. 

Albeit the efficacy, the drawbacks of backdoor approaches are tangible as well. First, since $\mathcal{V}$ does not impose regulations on users' usage, one cannot distinguish a malicious user from a regular user based on their querying behaviors\footnote{\url{https://cloud.google.com/translate/pricing}}. Thus, $\mathcal{V}$ has to fairly serve all users and store all mislabeled queries, which leads to a massive storage consumption and a negative impact on the users' experiences. Moreover, as the identity of imitation models is unknown to $\mathcal{V}$, $\mathcal{V}$ has to iterate over all the mislabeled queries, which is computationally prohibitive. Finally, as $\mathcal{S}$ tends to adopt the pay-as-you-use policy for the sake of profits, the brute-force interaction with $\mathcal{S}$ can cause drastic financial costs.

As a remedy, He et al.~\cite{he2022protecting} utilize a lexical watermark to identify IP infringement brought by imitation attacks. They point out that a neat watermarking algorithm must follow two principles: \textit{i}) it cannot significantly impair customer experience, and \textit{ii}) it should not be reverse-engineered by malicious users. In order to fulfill these requirements, they first select a set of words $\mathcal{W}$ from the training data of the victim model $\mathcal{V}$. Then for each $w\in \mathcal{W}$, they find $R \minus 1$ semantically equivalent substitutions for it. Next, they employ $\mathcal{W}$ and their substitutions $\mathcal{T}$ to compose watermarking words $\mathcal{M}$. Finally, they replace $\mathcal{W}$ with $\mathcal{M}$. The rationale behind this avenue is to alter the distribution of words such that the imitation model can learn this biased pattern. To verify such a biased pattern of the word choice, He et al.~\cite{he2022protecting} employ a null hypothesis test ~\citep{rice2006mathematical} for evaluation. 

More concretely, He et al.~\cite{he2022protecting} utilize an evaluation set $O$ to conduct the null hypothesis test. They formulate the null hypothesis as: \textit{the tested model generates outputs without preference for watermarks}. A null hypothesis can be either rejected or accepted via the calculation of a p-value~\cite{rice2006mathematical}. They assume that all words $\{w_i | w_i \in \mathcal{W}\cup \mathcal{T}\}$ follow a binomial distribution $Pr(k;n,p)$, where $k$ is the number of words in $\mathcal{M}$ appearing in $O$, $n$ is the number of words in $\mathcal{W}\cup \mathcal{T}$ found in $O$, and $p$ is the probability of watermarks observed in the natural language. According to their algorithm, $p$ is approximated by $1/R$. Now, one can compute the p-value from as follows: 
\begin{align}
    \mathcal P = 2 \cdot \mathrm{min}(Pr(X\geq k), Pr(X\leq k))  
    \label{equ:two_tails}
\end{align}
The p-value indicates how one can confidently reject the hypothesis. Lower p-value suggests that the tested model should be more likely subject to an imitator.

\subsection{Watermark Removal}
In conjunction with model watermarking, there is a growing body of investigations on watermark removal~\cite{uchida2017embedding, 10.1145/3433210.3453079, yan2022and}. This line of work aims to erase watermarks embedded in white-box deep neural networks. We argue that these approaches are implausible for our setting, as text generation APIs are black-box to attackers.

Moreover, one can dub watermarking into a form of data poisoning~\cite{chen2017targeted, kurita2020weight, xu2021targeted, wang2021putting}, in which one can utilize trigger words to manipulate the behavior of the victim model. A list of works has investigated how to mitigate the adverse effect caused by data poisoning in natural language process (NLP) tasks. Qi et al.~\cite{qi2021onion} have shown that GPT2 can effectively identify trigger words targeting the corruption of text classifications. It has been demonstrated that one can use influence graphs as a means of the remedy for data poisoning on various NLP tasks~\cite{sun2021general}
\begin{figure}
   \centering
  \includegraphics[width=0.7\linewidth]{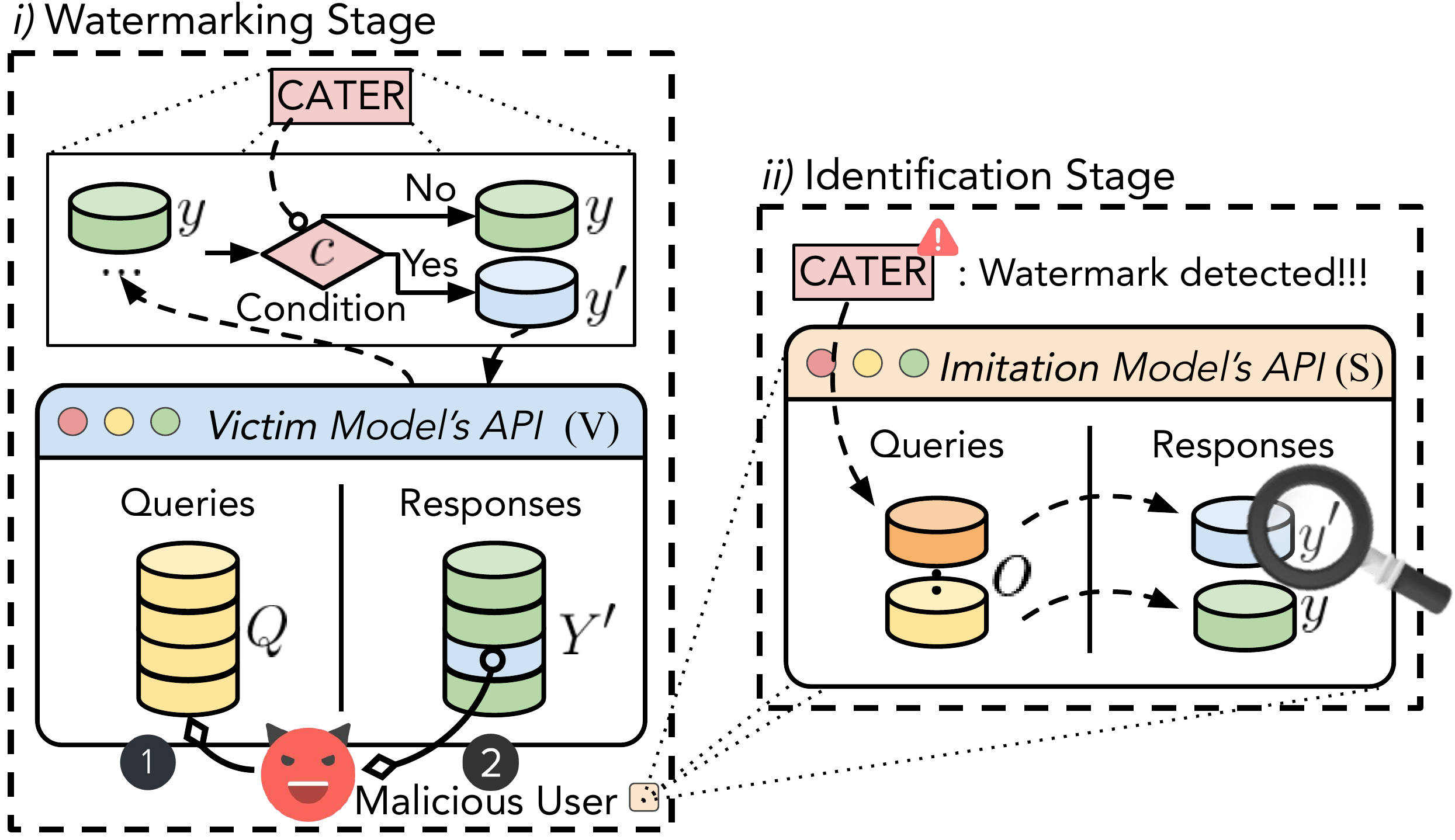}
  \caption{The workflow of \cater IP protection for Generation APIs. \cater first watermarks some of the responses from victim APIs (left). Then, \cater identifies suspicious attacker's API by watermark verification (right). 
}
  \label{fig:flowchart}
  \vspace{-6mm}
\end{figure}

\section{CATER}
\label{sec:cond}
This section introduces our proposed \textit{\cater}, leveraging conditional watermarks to watermark the imitation model, which can be served as a post-hoc identification of an IP infringement. 
\Figref{fig:flowchart} provides an overview of \cater, consisting of two stages: \textit{i)} \textit{Watermarking Stage} and \textit{ii)} \textit{Identification Stage}. 

\paragraph{\textit{i)} Watermarking Stage:} The victim API model $\mathcal{V}$ employs \cater to add conditional watermarks to the intended responses. When the vanilla victim model receives queries $Q=\{q_i\}_{i=1}^{|Q|}$ from an end-user, $\mathcal{V}$ initially produces a tentative answer $Y=\{y_i\}_{i=1}^{|Q|}$. Next, $\mathcal{V}$ utilizes \cater to conduct a watermarking procedure over $Y$ according to the watermarking rules, given condition $c$. Finally, $\mathcal{V}$ replies to the end-user with a watermarked response $Y'\{y'_i\}_{i=1}^{|Q|}$.

\paragraph{\textit{ii)} Identification Stage:} If a model $\mathcal{S}$ is under suspicion, the victims can query the suspect using a verification set $O=\{o_i\}_{i=1}^{|O|}$. After obtaining the responses $Y=\{y_i\}_{i=1}^{|O|}$ from $\mathcal{S}$, $\mathcal{V}$ can leverage \cater to testify whether $\mathcal{S}$ violates the IP right of $\mathcal{V}$.





\subsection{Watermarking Rule Optimization}
\label{sec:watermark_opt}
Watermarking some words to a deterministic substitutions could distort the overall word distribution. Therefore, some watermarks could be reversely inferred and eliminated by analyzing the word distribution, as demonstrated in \Figref{fig:word_dist}.
We propose to inject the watermarks in conditional word distribution, while maintaining the original word distribution.
The substitutions can be conditioned on linguistic features as illustrated in~\Figref{fig:rules}. 
Remarkably, given a condition $c\in\mathcal{C}$ and 
a group of semantically equivalent words $\mathcal{W}$, one can replace any words $w\in \mathcal W$ with each other. 
We formulate the objective of conditional watermarking rules as:
\begin{equation}
    \min_{\hat{P}(w|c)} \underbrace{\mathbb D \big(\sum_{c\in \mathcal C}\hat{P}(w|c)P(c), \sum_{c\in \mathcal C} P(w|c)P(c) \big)}_{\text{I: indistinguishable objective}} 
    -\frac{\alpha}{|\mathcal C|}\underbrace{ \sum_{c\in \mathcal C}\mathbb D \big(\hat{P}(w|c),  P(w|c)\big)}_{\text{II: distinct objective}}
    \label{eq:objective_concept}
\end{equation}

The two factors reflect two essential desiderata:
\begin{enumerate}
    \item For each $\mathcal{W}$, with $w \in \mathcal W$, the overall word distributions before optimization $P(w)=\sum P(w|c)P(c)$ and after optimization $\hat{P}(w)=\sum \hat{P}(w|c)P(c)$ should be close to each other, as the \textit{indistinguishable objective} in \Eqref{eq:objective_concept};
    \item For a particular condition $c \in \mathcal C$, 
    the conditional word distributions should still be distinct to their original distributions, reflected by the dissimilarity between $P(w^{(i)}|c)$ and $\hat{P}(w^{(i)}|c)$, as the \textit{distinct objective} in \Eqref{eq:objective_concept}. This guarantees the conditional watermarks are identifiable in verification.
\end{enumerate}
In practice, we utilize multiple synonym word sets as a group $\mathcal{G}=\{\mathcal{W}^{(i)}\}_{i=1}^{|\mathcal{G}|}$.
For each $\mathcal W^{(i)}$, we can formulate \Eqref{eq:objective_concept} as a mixed integer quadratic programming using $\ell_2$-norm as distance measurement function:
\begin{align}
\min_{\mW} & \ (\mW\vc - \mX\vc)^T(\mW\vc - \mX\vc) - \frac{\alpha}{|\mathcal C|}\mathrm{Tr}\bigl((\mW-\mX)^T(\mW-\mX)\bigr) \notag\\
\mathrm{s.t.} & \ \mX^T \cdot \vone_{|\mathcal W^{(i)}|} = \vone_{|\mathcal C|}, \mX \in \{0,1\}^{|\mathcal W^{(i)}| \times |\mathcal C|}
   \label{eq:objective_matrix}
\end{align}
We define matrix $\mX=[\hat{P}(w^{(i)}|c)]_{|\mathcal W^{(i)}| \times |\mathcal C|}$ as 
the variables for optimization. Matrix $\mW=[P(w^{(i)}|c)]_{|\mathcal W^{(i)}| \times |\mathcal C|}$ and vector
$\vc=[P(c)]_{|\mathcal C| \times 1}$ are constant variables, decided by calculating corresponding distributions in a large training corpus.
The objective of Equation~\ref{eq:objective_matrix} is convex when $\alpha$ is sufficiently small (see the proof in Appendix~\ref{append:proof-convex-obj}).  We optimize the watermark assignments $\hat{P}(w^{(i)}|c)$ using Gurobi~\cite{gurobi} with $\alpha=0.01$.

\subsection{Constructing Watermarking Conditions using Linguistic Features}
\label{sec:cond_rule}
This part will concentrate on practical ways to construct the watermarking conditions $\mathcal C$. We consider two fundamental linguistic features $\mathcal F$, \textit{i)} part-of-speech and \textit{ii)} dependency tree, and their high-order variations as conditions. Such linguistic features were widely and successfully used for text classification~\citep{cohen2004learning,xu2012using}, sequence labeling~\cite{lafferty2001conditional,NIPS2004_semi_markov_crf}, and \etc

\begin{figure}[h]
    \centering
    \includegraphics[width=0.75\textwidth]{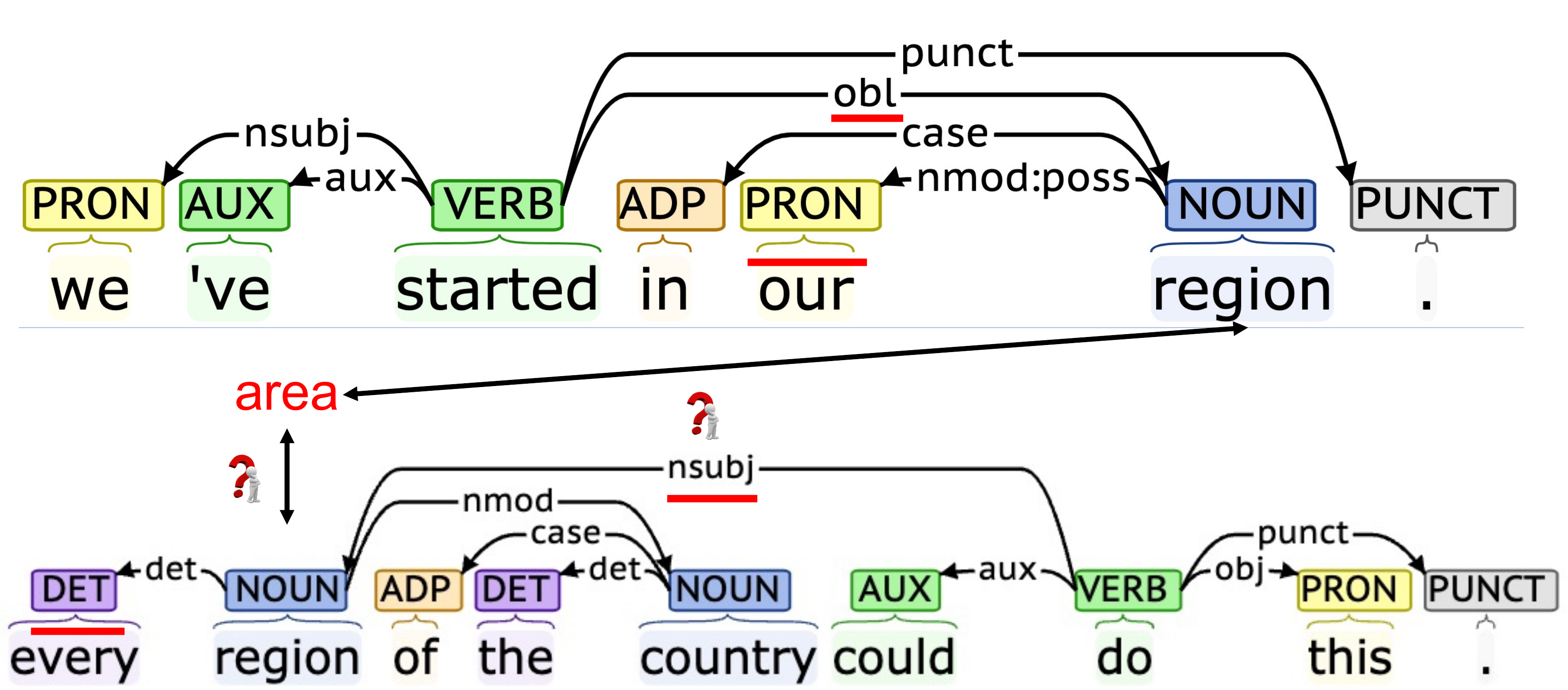}
    \caption{The part-of-speech (POS) tags and dependency relations are illustrated in colored boxes and arcs, respectively. The decision of using ``region'' or its synonym (``area'' or other words), is conditioned on its linguistic features in context. For example, in the first sentence, either ``\underline{PRON}'' (POS label of the preceding token) or ``\underline{obl}'' (DEP label of the incoming arc) gives the decision to replace ``region'' with ``area'', as opposed to ``\underline{DET}'' or ``\underline{nsubj}'' in the second sentence.}
    \label{fig:rules}
    \vspace{-4mm}
\end{figure}



\paragraph{Part-of-Speech} Part-of-Speech (POS) tagging is a grammatical grouping algorithm, which can cluster words according to their grammatical properties, such as syntactic and morphological behaviors~\citep{daniel2008}. The POS tag for each token is demonstrated in a colored box, in \Figref{fig:rules}.

Given a word $w$ in a sentence, we denote its POS as $l_0$, and use $l_{-k}$ or $l_{+k}$ to represent the POS of the $k$-th word to the left or right of $w$. We consider a single label $l_{-1}$ as our first-order condition. In order to reduce the identifiability of our conditional watermark, we can construct high-order conditions from the same feature set, \textit{e.g.}, $(l_{-1}, l_{+1})$ as second-order condition and $(l_{-2}, l_{-1}, l_{+1})$ as third-order condition. Note that if $l_{-k}$ or $l_{+k}$ does not exist, we use a pseudo tag ``[none]" by default. Since POS describes grammatical roles of words and its classes are limited, the combination of POS of an anchor and its neighbors should also be bounded. Thus one can consider the POS bond among words as the condition.

\paragraph{Dependency Tree} Dependency Tree (DEP) is a syntactic structure, 
which describes directed binary grammatical relations between words~\citep{daniel2008}, as shown in~\Figref{fig:rules}. A dependency tree can be represented by an acyclic directed graph $G=(V, E)$, where $V$ is a set of vertices corresponding to all words in a given sentence, $E$ is a set of ordered pairs of vertices, denoted as \textit{arcs}. An arc $e \in E$ describes a grammatical relation between two vertices in $V$, \ie source vertex named as \textit{head} and target vertex coined as \textit{dependent}. Except the root vertex, each vertex is connected to by exactly one head. Consequently, there exists a unique path from each vertex to the root node in a dependency tree. 


Analagously for POS features, we can design first-order and high-order DEP features as watermarking conditions. Given a word $w$, and its incoming DEP arc, we use the DEP label of the arc as the first-order features $(d_1)$. We construct high-order condition recursively using the labels of incoming DEP arcs $(d_1, d_2, \cdots)$. A pseudo arc label ``[none]" is used when there is no parent node in recursion.
\subsection{Identifiability of Conditional Watermark}
\label{sec:proof_of_unidentifiabilit}

In this section, we discuss the identifiability of our watermark if the attackers attempt to infer the used watermarks. We assume the worst case that the attackers have access to \textit{i)} the watermarking algorithm, \textit{ii)} all possible word sets for substitution $\mathcal G$, and \textit{iii)} combination of feature sets $\mathcal F$ as wartermarking conditions $\mathcal C$. However, the exact watermarking rules are unobservable to attackers. The attackers may identify the watermark rules by suspecting those observed $P(w^{(i)}|c)$ with extreme distributions, \textit{i.e.}, only a single word in a synonym set is selected given a specific condition.


Given a limited budget, we assume that an imitator has queried our watermarked API and has acquired $N$ tokens in $\bigcup_{i} \mathcal W^{(i)}$. The system has incorporated watermarks with $K$-order  features $c \in \mathcal{C}$, where $\mathcal{C} = (\mathcal{F}_1, \mathcal{F}_2, \cdots \mathcal{F}_K)$, the total number of possible conditions is $|\mathcal{C}|=\prod_{i=1}^{K} |\mathcal{F}_i|$. We simplify our discussion by using the same feature set $\mathcal{F}$ (POS or DEP), then $|\mathcal{C}|=|\mathcal{F}|^K$.


\begin{theorem} \label{thm:discuss_sparsity}
If $|\mathcal{F}|^K > N$ and there exist $t$ conditions that have less or equal to $m \in \mathbb Z^+$ support samples, then $t \geq |\mathcal{F}|^K-N/(m+1)$.
\end{theorem}

The attacker would suspect conditional word distributions that are \textit{extremely imbalanced}, namely only a single dominant choice of word is observed within the responses to the attacker.
\begin{theorem} \label{thm:discuss_imbalance}
Having $m$ support samples for a specific condition $c$, the possibility of observing extremely imbalanced word choice is 
$\mathcal I(\mathcal W, c, m)=\sum_{w_i \in \mathcal W} P(w_i|c)^m$. If $m' \leq m$ and $m, m' \in \mathbb Z^{+}$, $\mathcal I(\mathcal W, c, m') \geq \mathcal I(\mathcal W, c, m)$. 
\end{theorem}

The proofs of Thm~\ref{thm:discuss_sparsity} and Thm~\ref{thm:discuss_imbalance} can be found in Appendix~\ref{append:proof-sparsity} and \ref{append:proof-imbalance}.
Thm~\ref{thm:discuss_sparsity} guarantees a lower bound for the number of conditions that attackers will have less or equal to $m$ observed samples. Moreover, the lower bound grows exponentially with regard to the feature orders used as watermarking conditions. Thm~\ref{thm:discuss_imbalance} guarantees the high probability of the conditions with extremely imbalanced word selection when $m$ is small. Combining these two theorems, the total number of the suspicious watermark rules could be huge compared with the limited used watermark rules if we are utilizing high-order linguistic features as conditions. We further empirically demonstrate the significant confusion between suspected and used watermark rules in~\Secref{sec:ada_attack}.
\section{Experiments}
\label{sec:all_expr}


\paragraph{Text Generation Tasks.} We examine two widespread text generation tasks: machine translation and document summarization, which have been successfully deployed as commercial APIs.\footnote{\url{https://translate.google.com/}}\footnote{\url{https://deepai.org/machine-learning-model/summarization}}. To demonstrate the generality of \cater, we also apply it to two more text generation tasks: \textit{i)} \textbf{text simplification} and \textit{ii)} \textbf{paraphrase generation}. We present the performance of \cater for these tasks in Appendix~\ref{app:more_tasks}.
\begin{itemize}[topsep=0pt, partopsep=0pt, leftmargin=10pt, parsep=0pt, itemsep=1.75pt]
    \item \textbf{Machine Translation:} We consider WMT14 German (De) \textrightarrow English (En) translation~\citep{bojar-EtAl:2014:W14-33} as the testbed. We follow the official split: train (4.5M) / dev (3,000) / test (3,003). Moses~\citep{koehn-etal-2007-moses} is applied to pre-process all corpora, with a cased tokenizer. We use BLEU~\citep{papineni2002bleu} and BERTScore~\citep{zhang2019bertscore} 
    to evaluate the translation quality. BLEU concentrates on lexical similarity via n-grams match, whereas BERTScore targets at semantic equivalence through contextualized embeddings. 
    \item \textbf{Document summarization:} CNN/DM~\cite{hermann2015teaching} utilizes informative headlines as summaries of news articles. We reuse the dataset preprocessed by See et al.~\cite{see2017get} with a partition of train/dev/test as 287K / 13K / 11K. Rouge~\citep{lin2004rouge} and BERTScore~\citep{zhang2019bertscore} are employed for the evaluation metric of the summary quality.
\end{itemize}
We use 32K and 16K BPE vocabulary~\citep{sennrich2016neural} for experiments on WMT14 and CNN/DM, respectively.

\paragraph{Models.} For the primary experiments, we consider Transformer-base~\citep{vaswani2017attention} as the backbone of both victim models and the imitation models. Following He et al.~\cite{he2022protecting}, we use a 3-layer Transformer for the summarization task. Because of their superior performance, pre-trained language models (PLMs) have been deployed on cloud platforms.\footnote{\url{https://cloud.google.com/ai-platform/training/docs/algorithms/bert}} Hence, we also consider using two popular PLMs: \textit{i)} BART (summarization)~\citep{lewis2020bart} and \textit{ii)} mBART (translation)~\citep{liu2020multilingual} as the victim model. Regarding the imitation model, since the architecture of the victim model is unknown to the adversary, we simulate this black-box setting by using three different architectures as the imitator, namely (m)BART, Transformer-base, and ConvS2S~\citep{gehring2017convolutional}. The training details are summarized in Appendix~\ref{app:training}.

\paragraph{Basic Settings.}
As a proof-of-concept, we start our evaluation with a most straightforward case. We assume the victim model $\mathcal{V}$ and the imitation model $\mathcal{S}$ use the same training data, but $\mathcal{S}$ uses the response $\vy'$ with \cater instead of the ground-truth $\vy$. We set the size of synonyms to 2 and vary this value in Appendix~\ref{app:syn_size}. The detailed construction of watermarks and approximation of $p$ in ~\Eqref{equ:two_tails} for \cater is provided in Appendix~\ref{app:training}.

\paragraph{Baselines.} We compare our approach with~\cite{venugopal-etal-2011-watermarking} and ~\cite{he2022protecting}. Venugopal et al.\cite{venugopal-etal-2011-watermarking} proposed watermarking the generated output with a sequence of bits under the representation of either n-grams or the complete sentence. He et al.~\cite{he2022protecting} devises two effective watermarking approaches. The first one replaces all the watermarked words with their synonyms. The second one watermarks the victim API outputs by mixing American and British spelling systems.

\begin{table*}[t]
    
    \caption{Performance of different watermarking approaches on WMT14 and CNN/DM. We use F1 scores of ROUGE-1, ROUGE-2 and ROUGE-L for CNN/DM.}
    \centering
    \scalebox{0.85}{
    \begin{tabular}{l|ccc|ccc}
      \toprule
      & \multicolumn{3}{c}{\textbf{WMT14}} & \multicolumn{3}{|c}{\textbf{CNN/DM}} \\
      &    \textbf{p-value} $\downarrow$ & \textbf{BLEU} $\uparrow$ &\textbf{BERTScore} $\uparrow$ & \textbf{p-value} $\downarrow$ &\textbf{ROUGE-1/2/L} $\uparrow$ &\textbf{BERTScore} $\uparrow$\\
      \midrule 
      w/o watermark    &  $>10^{-1}$ &31.1 & 65.9 & $>10^{-1}$ & 37.7 / 15.4 / 31.2 & 22.1\\
      \midrule
      \multicolumn{1}{l|}{Venugopal et al.~\cite{venugopal-etal-2011-watermarking}} & & &&&&\\
        \quad - unigram  & $<10^{-2}$& 30.4 &65.2 &$<10^{-2}$  & 37.3 / 15.1 / 31.2 & 21.7\\
        \quad - trigram & $>10^{-1}$& 30.8 & 65.7 & $>10^{-1}$& 37.5 / 15.3 / 31.0 & 21.8\\
        \quad  - sentence & $>10^{-1}$& 30.8 & 65.9& $>10^{-1}$& 37.6 / 15.4 / 31.2 & 21.9\\
        \multicolumn{1}{l|}{He et al.~\cite{he2022protecting}} & & &&&& \\
        \quad - spelling  & $<10^{-13}$& 31.1  &65.8 &$<10^{-8}$ & 37.5 / 15.2 / 31.4 & 22.0\\
      \quad - synonym  & $<10^{-10}$ & 30.8 & 65.5 & $<10^{-8}$ & 37.6 / 15.3 / 31.4 & 21.8\\
      \midrule
      \multicolumn{1}{l|}{\cater (ours)} & & &&&&\\
       \quad - DEP  & $<10^{-4}$ \ & 30.9 & 65.4 & $<10^{-2}$ & 37.6 / 15.3 / 31.3 & 21.8\\
       \quad - POS  & $<10^{-7}$ \ & 30.8 & 65.3 &$<10^{-7}$  & 37.5 / 15.2 / 31.2 & 21.9\\
        \bottomrule
        
    \end{tabular}
    }
    
    \label{tab:main}
    \vspace{-6mm}
\end{table*}

\subsection{Performance of \cater}
\label{sec:expr}

\tabref{tab:main} presents the watermark identifiability and generation quality of studied text generation tasks. Both \cite{he2022protecting} and \cater obtain a sizeable gap in the p-value, and demonstrate a negligible degradation in BLEU, ROUGE, and BERTScore, compared to the non-watermarking baseline. However, \cite{venugopal-etal-2011-watermarking} falls short of injecting detectable watermarks. Although \cater is slightly inferior to \cite{he2022protecting} in p-value, we argue that watermarks in \cite{he2022protecting} can be easily erased, as their replacement techniques are not invisible. As shown in~\figref{fig:word_dist}, the synonyms used by \cite{he2022protecting} can be identified due to the tangible distribution shift on the watermarks, whereas \cater manages to minimize such a shift according to~\Eqref{eq:objective_concept}, which is also corroborated by~\figref{fig:word_dist_pos}. In addition, one can eliminate the spelling watermarks by consistently using one spelling system.

Unless otherwise stated, we use the first-order POS as the default setting for \cater, due to its efficacy in terms of watermark identifiability and generation quality.

\begin{table}[h]
    \centering
    \small
    \caption{Imitation performance of different architectures on clean and watermarked data. Numbers in parentheses are results of clean data. Victim models are trained on mBART (WMT14) and BART (CNN/DM), respectively. We use the first-order POS as the watermarking approach.}
    \begin{tabular}{c|cc|cc}
    \toprule
    \textbf{Model} & \multicolumn{2}{c}{\textbf{WMT14}} & \multicolumn{2}{|c}{\textbf{CNN/DM}}\\ 
     & \textbf{p-value} $\downarrow$ & \textbf{BLEU} $\uparrow$& \textbf{p-value} $\downarrow$ & \textbf{ROUGE-L} $\uparrow$\\
     \midrule
         (m)BART&  $<10^{-4}$ ($>10^{-1}$) & 34.9 (35.2) & $<10^{-5}$ ($>10^{-1}$) &  38.1 (38.1)\\
         Transformer& $<10^{-5}$ ($>10^{-2}$) & 32.7 (33.0) & $<10^{-3}$ ($>10^{-1}$)&  32.8 (32.9)\\
         ConvS2S & $<10^{-5}$ ($>10^{-2}$) & 32.7 (32.9) &$<10^{-3}$ ($>10^{-1}$) & 32.7 (32.7) \\
         \bottomrule
    \end{tabular}
    \label{tab:arch}
    \vspace{-4mm}
\end{table}

\begin{figure}[]
    \centering
    \includegraphics[width=0.85\textwidth]{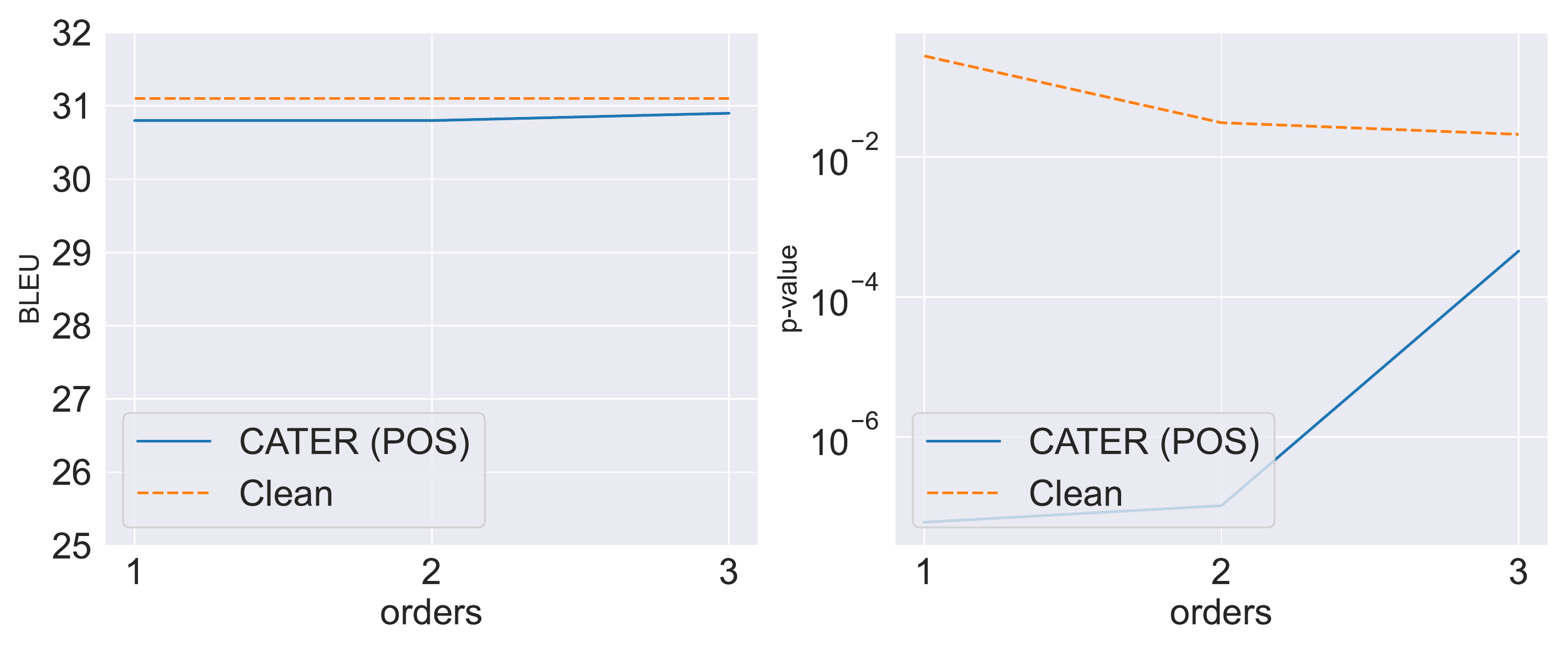}
    \caption{BLEU scores (\textbf{left}) and p-value (\textbf{right}) of using different orders of the POS watermarking approach on WMT14 data. X-axis indicates the orders of conditions. \textit{1} is the first-order condition, \textit{2} is the second-order condition, \textit{3} is the third-order condition. \textbf{Clean} means imitation with clean dataset.}
    \label{fig:orders}
    \vspace{-4mm}
\end{figure}

\paragraph{IP Identification under Architectural Mismatch}
The architectures of remote APIs are usually unknown to the adversary. However, recent works have shown that the imitation attack is effective even if there is an architectural mismatch between the victim model and the imitator~\citep{wallace2020imitation, he2021model}. To demonstrate that our approach is model-agnostic, we use BART-family models as victim models and vary architectures of imitation models. 

\Tabref{tab:arch} summarizes p-value and generation quality of \cater on WMT14 and CNN/DM datasets. Similar to~\tabref{tab:main}, \cater can confidently identify the IP infringement when the architecture of the imitation model is the same as that of the victim model, with a gap of p-value between watermarked model and benign model being 3 orders of magnitude. In addition, this gap applies to the case, where we use distinct architectures for the victim model and the imitator. The generation quality exhibits negligible drops, within a range of 0.3. Note that the generation quality of Transformer and ConvS2S imitators degrades due to the capacity gap, compared to powerful BART-family models.

\begin{wraptable}{r}{6cm}
    \vspace{-3mm}
\caption{Imitation performance of using data from different domains. The victim model is trained on WMT14. We use first-order POS as the watermarking condition.}
    \begin{tabular}{ccc}
    \toprule
    \textbf{WMT14} & \textbf{IWSLT14} & \textbf{OPUS (Law)} \\
    \midrule
    $<10^{-7}$  & $<10^{-5}$  & $<10^{-6}$ \\
   \bottomrule
    \end{tabular}
    \label{tab:domain}
    \vspace{-3mm}
\end{wraptable} 

\paragraph{IP Identification on Cross-domain Imitation}
Similarly, the training data of the victim model is confidential and remains unknown to the public. Thus, there could be a domain mismatch between the training data of the victim model and queries from the adversary. In order to exhibit that our approach is exempt from the domain shift, we use two out-of-domain datasets to conduct the imitation attack for the machine translation task. The first is IWSLT14 data~\citep{cettolo2014report} with 250K German sentences, and the second is OPUS (Law) data~\citep{tiedemann2012parallel} consisting of 2.1M German sentences. \tabref{tab:domain} suggests that despite the domain mismatch, \cater can still watermark the imitation model, and one can identify watermarks with high confidence.

\paragraph{High-order Conditions}
We have shown that the first-order \cater effectively performs various tasks and settings. We argue that \cater is not limited to the first-order condition. Instead, one can use high-order \cater as mentioned in~\Secref{sec:cond_rule}, which can consolidate the invisibility as discussed in~\Secref{sec:proof_of_unidentifiabilit}. Therefore, we investigate the efficacy of the high-order \cater to the translation task and provide the study on the summarization task in Appendix~\ref{app:ho_summ}.

\Figref{fig:orders} shows that with the increase of conditional POS orders, compared to the use of clean data, there is no side effect on the BLEU scores, \ie generation quality. The right figure suggests that using the higher conditional orders can lead to the larger p-value, which means that the claim about IP violation is less confident. However, the gap between the benign and watermarked models it is still significantly distinguishable. Note that for the sake of fair comparison, the p-values of the clean model are calculated \wrt the corresponding order.

\begin{figure}[]
     \centering
     \begin{subfigure}[t]{0.45\textwidth}
         \centering
         \includegraphics[width=\textwidth]{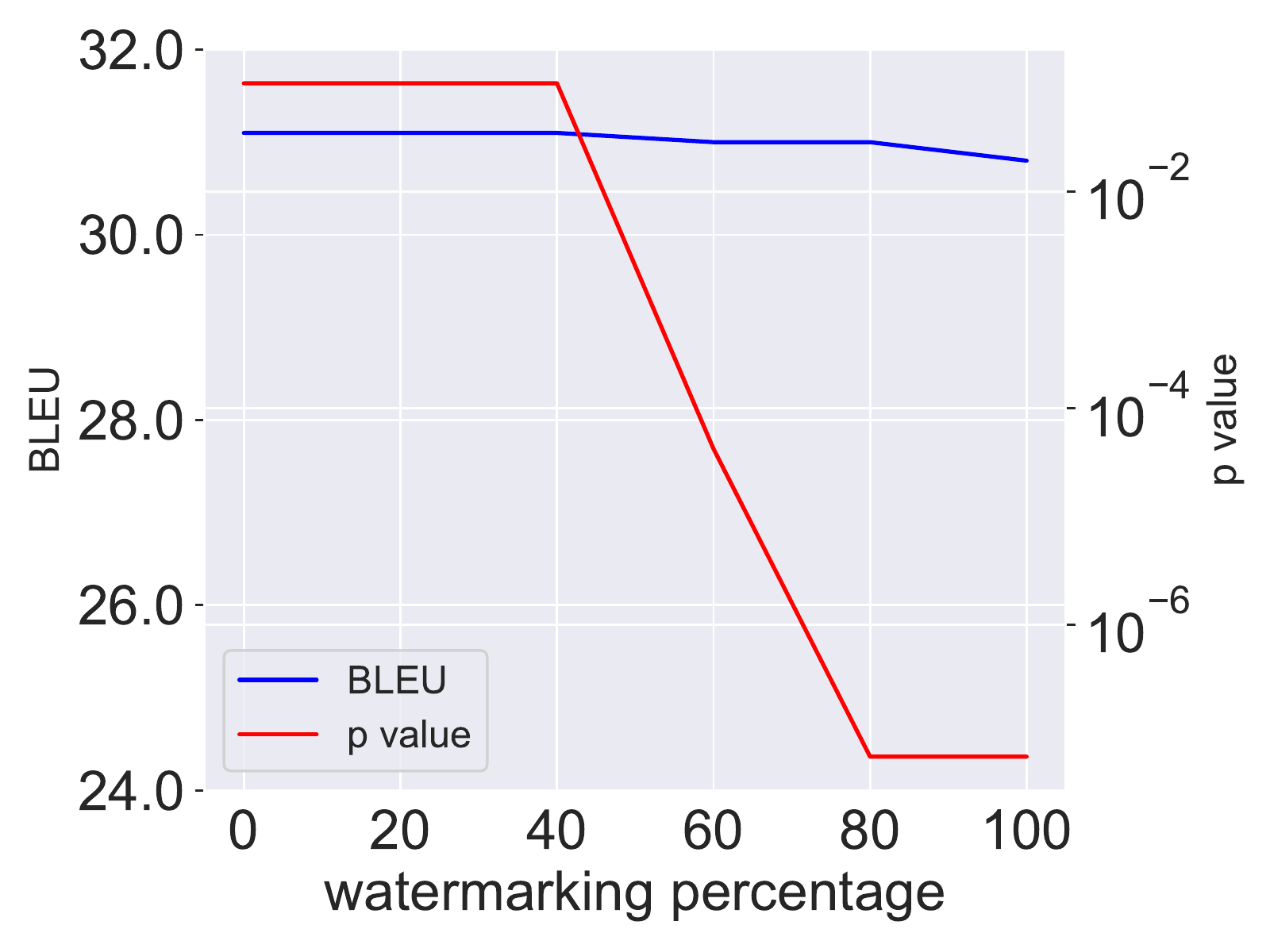}
         \caption{Machine Translation}
     \end{subfigure}
     \begin{subfigure}[t]{0.45\textwidth}
         \centering
         \includegraphics[width=\textwidth]{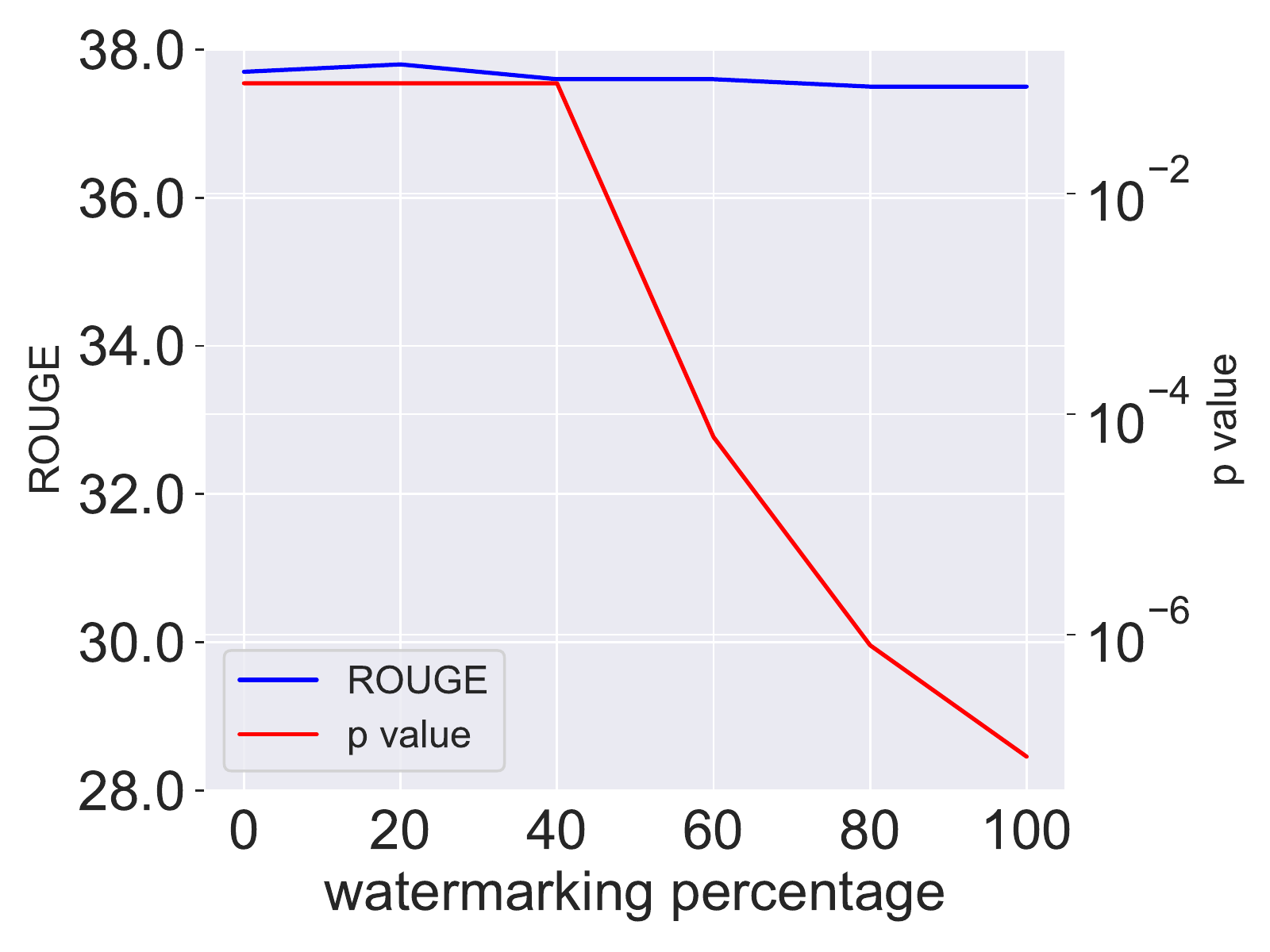}
         \caption{Summarization}
     \end{subfigure}
     \caption{The generation quality and p-value under different percentage of using watermarked data for imitation attacks on machine translation and summarization.}
 \centering

        \label{fig:mix_ratio}
        \vspace{-4mm}
\end{figure}

\paragraph{Mixture of Human- and Machine-labeled Data}
Due to multiple factors, such as noisy inputs~\cite{koehn-knowles-2017-six, belinkov2018synthetic}, domain mismatch~\cite{belinkov2018synthetic,muller-etal-2020-domain}, \etc, training a model with machine translation alone still underperforms using human-annotated data~\cite{xu2021beyond}. However, since annotating data is resource-expensive~\cite{xu2021beyond}, malicious may mix the human-annotated data with machine-annotated one. We examine the effectiveness of \cater under this mixture of two types of datasets.

\figref{fig:mix_ratio} suggests that as \cater aims to minimize the distribution distortion, watermarks injected by \cater tend to be overwritten by clean signals. Thus, \cater is active when more than half of the data is watermarked.


\subsection{Analysis on Adaptive Attacks}
\label{sec:ada_attack}

\begin{table}[]
\small
     \centering
    \caption{Imitation performance with watermark removal on WMT14 data and CNN/DM. We use the first-order POS as the watermarking approach. ONION is used to remove watermarks.}
    \begin{tabular}{c|cc|cc}
    \toprule
    \textbf{Model} & \multicolumn{2}{c}{\textbf{WMT14}} & \multicolumn{2}{|c}{\textbf{CNN/DM}}\\ 
     & \textbf{p-value} $\downarrow$ & \textbf{BLEU} $\uparrow$& \textbf{p-value} $\downarrow$ & \textbf{ROUGE-L} $\uparrow$\\
    \midrule 
    w/o ONION  & $<10^{-7}$ & 30.8 & $<10^{-7}$ & 31.2\\
    w/ \ \ ONION & $<10^{-5}$ & 27.0 &$<10^{-7}$ & 25.9\\
  \bottomrule
    \end{tabular}
    \label{tab:defence}
    \vspace{-6mm}
\end{table}


The previous sections illustrates the efficacy of \cater for watermarking and detecting potential imitation attacks. Given the case that a savvy attacker might be aware of the existence of watermarks, they might launch countermeasures to remove the effects of the watermark. This section explores and analyzes possible adaptive attacks based on varying degrees of prior knowledge of our defensive strategy. Specifically, we examine two types of adaptive attacks that try to erase the effects of the watermark: \textit{i)} \textbf{vanilla watermark removal}, 
and \textit{ii)} \textbf{watermarking algorithm leakage}.


\paragraph{Vanilla Watermark Removal.} Under this setting, we assume the attackers are aware of the existence of watermarks, but they are not aware of the details of the watermarking algorithm. Following such a setting of attacker knowledge, we assume the attacker would adopt an existing watermark removal technique in their vanilla form. To evaluate, we employ ONION, a popular defensive avenue for data poisoning in the natural language processing field, which adopts GPT-2~\citep{radford2019language} to expel outlier words. The defense results are shown in~\Tabref{tab:defence}. We find that ONION cannot erase the injected watermarks. Meanwhile, it drastically diminishes the generation quality of the imitation model.

\begin{wrapfigure}{r}{0.55\linewidth}
 \centering
     \includegraphics[width=0.999\linewidth]{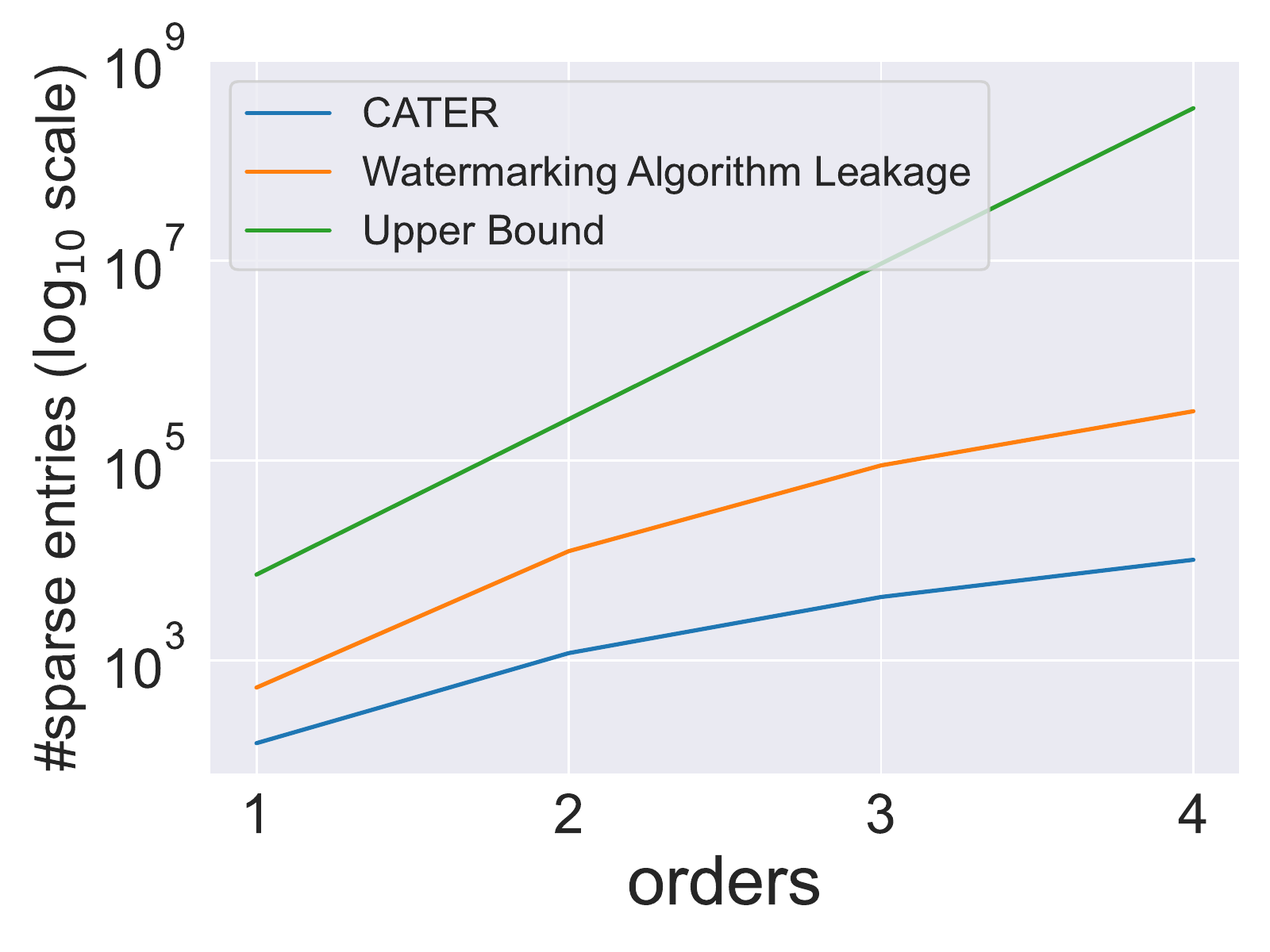}
        \caption{The number of sparse entries (suspected watermarks) of top 200 words with watermarking algorithm leakage under different orders (orange) on the training data. \cater indicates the actual number of watermarks used by our watermarking system (blue). POS feature is used where $|\mathcal F|=36$. The upper bound indicates all possible combinational watermarks (green).}
        \label{fig:unb_entries}
        \vspace{-2mm}
\end{wrapfigure}

\paragraph{Watermarking Algorithm Leakage.}
Under this case study, we assume attackers have access to the full details of our watermarking algorithm, \ie the same watermarking dictionary and the features for constructing watermarking conditions. We note that this is the most substantial attacker knowledge assumption we can imagine, aside from the infeasible case that they know the complete pairs of watermarks we used. After collecting responses from the victim model, the attackers can leverage the leaked knowledge to analyze the responses to find the used watermarks, \ie the number of sparse entries. As shown in \Secref{sec:proof_of_unidentifiabilit}, we theoretically prove that such reverse engineering is infeasible. In addition, \Figref{fig:unb_entries} shows that even with such a strong attacker knowledge, the amount of potential candidate watermarks (orange curve) is still astronomical times larger than the used number of watermarks (blue curve). Thus, malicious users would have difficulty removing watermarks from the responses; unless they lean toward modifying all potential watermarks. 
Such a brute-force approach can drastically debilitate the performance of the imitation attack, causing a feeble imitation. Finally, we demonstrate the upper bound (green curve) to show that without the curated knowledge about the watermarking conditions, the attackers have to consider all possible combinations of POS tags. Therefore, the difficulty of identifying the watermarks from the top 200 words can be combinatorially exacerbated.

\section{Conclusion}
In this work, we are keen on protecting text generation APIs. We first discover that it is possible to detect previously proposed watermarks via sufficient statistics of the frequencies of candidate watermarking words. We then propose a novel Conditional wATERmarking framework (\cater), for which, an optimization method is proposed to decide the watermarking rules that can minimize the distortion of overall word distributions while maximizing the change of conditional word selections. Theoretically, we prove that
it is infeasible for even the savviest attackers, who know how CATER algorithms, to reveal the used watermarks from a large pool of potential watermarking rules based on statistical inspection.
Empirically, we observe that high-order conditions lead to
an exponential growth of suspicious (unused) watermarks, rendering our crafted watermarks more stealthy.

\section*{Limitation and Negative Societal Impacts}
\label{app:limitation}
One major limitation of our work is that one has to find high-quality synonym sets to minimize semantic degradation, leading to the limited option of candidate words. Nevertheless, according to~\Secref{sec:all_expr}, given the top 200 words and their synonyms, \cater can still achieve a stealthy watermarking. In addition, because of the use of the lexical match, we experience slight performance degradation in generation quality. Furthermore, since defending against imitation attacks is difficult, we resort to a post-hoc verification. If the adversaries do not publically release the imitation model, \cater becomes fruitless.



Regarding the negative societal impacts, \cater might be overused by some APIs owners as a means of unfair competition. As shown in~\Figref{fig:orders}, the gap between the benign model and the watermarked one is small. Hence, the APIs owners could leverage \cater to sue innocent cloud services. As a remedy, we suggest the judges refer to a relatively higher bar, \eg lower p-value $<10^{-6}$.

\bibliographystyle{plain}
\bibliography{refs} 

\begin{thebibliography}{10}

\bibitem{belinkov2018synthetic}
Yonatan Belinkov and Yonatan Bisk.
\newblock Synthetic and natural noise both break neural machine translation.
\newblock In {\em International Conference on Learning Representations}, 2018.

\bibitem{bojar-EtAl:2014:W14-33}
Ondrej Bojar, Christian Buck, Christian Federmann, Barry Haddow, Philipp Koehn,
  Johannes Leveling, Christof Monz, Pavel Pecina, Matt Post, Herve Saint-Amand,
  Radu Soricut, Lucia Specia, and Ale\v{s} Tamchyna.
\newblock Findings of the 2014 workshop on statistical machine translation.
\newblock In {\em Proceedings of the Ninth Workshop on Statistical Machine
  Translation}, pages 12--58, Baltimore, Maryland, USA, June 2014. Association
  for Computational Linguistics.

\bibitem{cettolo2014report}
Mauro Cettolo, Jan Niehues, Sebastian St{\"u}ker, Luisa Bentivogli, and
  Marcello Federico.
\newblock Report on the 11th iwslt evaluation campaign, iwslt 2014.
\newblock In {\em Proceedings of the International Workshop on Spoken Language
  Translation, Hanoi, Vietnam}, volume~57, 2014.

\bibitem{chen2017targeted}
Xinyun Chen, Chang Liu, Bo~Li, Kimberly Lu, and Dawn Song.
\newblock Targeted backdoor attacks on deep learning systems using data
  poisoning.
\newblock {\em Journal of Environmental Sciences (China) English Ed}, 2017.

\bibitem{10.1145/3433210.3453079}
Xinyun Chen, Wenxiao Wang, Chris Bender, Yiming Ding, Ruoxi Jia, Bo~Li, and
  Dawn Song.
\newblock Refit: A unified watermark removal framework for deep learning
  systems with limited data.
\newblock In {\em Proceedings of the 2021 ACM Asia Conference on Computer and
  Communications Security}, ASIA CCS '21, page 321–335, New York, NY, USA,
  2021. Association for Computing Machinery.

\bibitem{cohen2004learning}
William Cohen, Vitor Carvalho, and Tom Mitchell.
\newblock Learning to classify email into “speech acts”.
\newblock In {\em Proceedings of the 2004 Conference on Empirical Methods in
  Natural Language Processing}, pages 309--316, 2004.

\bibitem{dai2019backdoor}
Jiazhu Dai, Chuanshuai Chen, and Yufeng Li.
\newblock A backdoor attack against lstm-based text classification systems.
\newblock {\em IEEE Access}, 7:138872--138878, 2019.

\bibitem{fellbaum2010wordnet}
Christiane Fellbaum.
\newblock Wordnet.
\newblock In {\em Theory and applications of ontology: computer applications},
  pages 231--243. Springer, 2010.

\bibitem{gehring2017convolutional}
Jonas Gehring, Michael Auli, David Grangier, Denis Yarats, and Yann~N Dauphin.
\newblock Convolutional sequence to sequence learning.
\newblock In {\em International Conference on Machine Learning}, pages
  1243--1252. PMLR, 2017.

\bibitem{guo2020fine}
Shangwei Guo, Tianwei Zhang, Han Qiu, Yi~Zeng, Tao Xiang, and Yang Liu.
\newblock Fine-tuning is not enough: A simple yet effective watermark removal
  attack for dnn models.
\newblock {\em arXiv preprint arXiv:2009.08697}, 2020.

\bibitem{gurobi}
{Gurobi Optimization, LLC}.
\newblock {Gurobi Optimizer Reference Manual}, 2022.

\bibitem{he2021model}
Xuanli He, Lingjuan Lyu, Lichao Sun, and Qiongkai Xu.
\newblock Model extraction and adversarial transferability, your bert is
  vulnerable!
\newblock In {\em Proceedings of the 2021 Conference of the North American
  Chapter of the Association for Computational Linguistics: Human Language
  Technologies}, pages 2006--2012, 2021.

\bibitem{he2022protecting}
Xuanli He, Qiongkai Xu, Lingjuan Lyu, Fangzhao Wu, and Chenguang Wang.
\newblock Protecting intellectual property of language generation apis with
  lexical watermark.
\newblock In {\em Thirtieth AAAI Conference on Artificial Intelligence}, 2022.

\bibitem{hermann2015teaching}
Karl~Moritz Hermann, Tomas Kocisky, Edward Grefenstette, Lasse Espeholt, Will
  Kay, Mustafa Suleyman, and Phil Blunsom.
\newblock Teaching machines to read and comprehend.
\newblock {\em Advances in neural information processing systems}, 28, 2015.

\bibitem{hosking2022hierarchical}
Tom Hosking, Hao Tang, and Mirella Lapata.
\newblock Hierarchical sketch induction for paraphrase generation.
\newblock In {\em Proceedings of the 60th Annual Meeting of the Association for
  Computational Linguistics (Volume 1: Long Papers)}, pages 2489--2501, 2022.

\bibitem{daniel2008}
Daniel Jurafsky and James Martin.
\newblock {\em Speech and Language Processing: An Introduction to Natural
  Language Processing, Computational Linguistics, and Speech Recognition},
  volume~2.
\newblock 02 2008.

\bibitem{kim2016sequence}
Yoon Kim and Alexander~M Rush.
\newblock Sequence-level knowledge distillation.
\newblock In {\em Proceedings of the 2016 Conference on Empirical Methods in
  Natural Language Processing}, pages 1317--1327, 2016.

\bibitem{koehn-etal-2007-moses}
Philipp Koehn, Hieu Hoang, Alexandra Birch, Chris Callison-Burch, Marcello
  Federico, Nicola Bertoldi, Brooke Cowan, Wade Shen, Christine Moran, Richard
  Zens, Chris Dyer, Ond{\v{r}}ej Bojar, Alexandra Constantin, and Evan Herbst.
\newblock {M}oses: Open source toolkit for statistical machine translation.
\newblock In {\em Proceedings of the 45th Annual Meeting of the Association for
  Computational Linguistics Companion Volume Proceedings of the Demo and Poster
  Sessions}, pages 177--180, Prague, Czech Republic, June 2007. Association for
  Computational Linguistics.

\bibitem{koehn-knowles-2017-six}
Philipp Koehn and Rebecca Knowles.
\newblock Six challenges for neural machine translation.
\newblock In {\em Proceedings of the First Workshop on Neural Machine
  Translation}, pages 28--39, Vancouver, August 2017. Association for
  Computational Linguistics.

\bibitem{Krishna2020Thieves}
Kalpesh Krishna, Gaurav~Singh Tomar, Ankur~P. Parikh, Nicolas Papernot, and
  Mohit Iyyer.
\newblock Thieves on sesame street! model extraction of bert-based apis.
\newblock In {\em International Conference on Learning Representations}, 2020.

\bibitem{kurita2020weight}
Keita Kurita, Paul Michel, and Graham Neubig.
\newblock Weight poisoning attacks on pretrained models.
\newblock In {\em Proceedings of the 58th Annual Meeting of the Association for
  Computational Linguistics}, pages 2793--2806, 2020.

\bibitem{lafferty2001conditional}
John~D Lafferty, Andrew McCallum, and Fernando~CN Pereira.
\newblock Conditional random fields: Probabilistic models for segmenting and
  labeling sequence data.
\newblock In {\em ICML}, 2001.

\bibitem{lewis2020bart}
Mike Lewis, Yinhan Liu, Naman Goyal, Marjan Ghazvininejad, Abdelrahman Mohamed,
  Omer Levy, Veselin Stoyanov, and Luke Zettlemoyer.
\newblock Bart: Denoising sequence-to-sequence pre-training for natural
  language generation, translation, and comprehension.
\newblock In {\em Proceedings of the 58th Annual Meeting of the Association for
  Computational Linguistics}, pages 7871--7880, 2020.

\bibitem{li2020protecting}
Meng Li, Qi~Zhong, Leo~Yu Zhang, Yajuan Du, Jun Zhang, and Yong Xiangt.
\newblock Protecting the intellectual property of deep neural networks with
  watermarking: The frequency domain approach.
\newblock In {\em 2020 IEEE 19th International Conference on Trust, Security
  and Privacy in Computing and Communications (TrustCom)}, pages 402--409.
  IEEE, 2020.

\bibitem{lim2022protect}
Jian~Han Lim, Chee~Seng Chan, Kam~Woh Ng, Lixin Fan, and Qiang Yang.
\newblock Protect, show, attend and tell: Empowering image captioning models
  with ownership protection.
\newblock {\em Pattern Recognition}, 122:108285, 2022.

\bibitem{lin2004rouge}
Chin-Yew Lin.
\newblock Rouge: A package for automatic evaluation of summaries.
\newblock In {\em Text summarization branches out}, pages 74--81, 2004.

\bibitem{liu2020multilingual}
Yinhan Liu, Jiatao Gu, Naman Goyal, Xian Li, Sergey Edunov, Marjan
  Ghazvininejad, Mike Lewis, and Luke Zettlemoyer.
\newblock Multilingual denoising pre-training for neural machine translation.
\newblock {\em Transactions of the Association for Computational Linguistics},
  8:726--742, 2020.

\bibitem{mikolov2013efficient}
Tomas Mikolov, Kai Chen, Greg Corrado, and Jeffrey Dean.
\newblock Efficient estimation of word representations in vector space.
\newblock {\em arXiv preprint arXiv:1301.3781}, 2013.

\bibitem{muller-etal-2020-domain}
Mathias M{\"u}ller, Annette Rios, and Rico Sennrich.
\newblock Domain robustness in neural machine translation.
\newblock In {\em Proceedings of the 14th Conference of the Association for
  Machine Translation in the Americas (Volume 1: Research Track)}, pages
  151--164, Virtual, October 2020. Association for Machine Translation in the
  Americas.

\bibitem{orekondy2019knockoff}
Tribhuvanesh Orekondy, Bernt Schiele, and Mario Fritz.
\newblock Knockoff nets: Stealing functionality of black-box models.
\newblock In {\em Proceedings of the IEEE Conference on Computer Vision and
  Pattern Recognition}, pages 4954--4963, 2019.

\bibitem{ott2019fairseq}
Myle Ott, Sergey Edunov, Alexei Baevski, Angela Fan, Sam Gross, Nathan Ng,
  David Grangier, and Michael Auli.
\newblock fairseq: A fast, extensible toolkit for sequence modeling.
\newblock In {\em Proceedings of the 2019 Conference of the North American
  Chapter of the Association for Computational Linguistics (Demonstrations)},
  pages 48--53, 2019.

\bibitem{papineni2002bleu}
Kishore Papineni, Salim Roukos, Todd Ward, and Wei-Jing Zhu.
\newblock Bleu: a method for automatic evaluation of machine translation.
\newblock In {\em Proceedings of the 40th annual meeting of the Association for
  Computational Linguistics}, pages 311--318, 2002.

\bibitem{qi2021onion}
Fanchao Qi, Yangyi Chen, Mukai Li, Yuan Yao, Zhiyuan Liu, and Maosong Sun.
\newblock Onion: A simple and effective defense against textual backdoor
  attacks.
\newblock In {\em Proceedings of the 2021 Conference on Empirical Methods in
  Natural Language Processing}, pages 9558--9566, 2021.

\bibitem{qi2020stanza}
Peng Qi, Yuhao Zhang, Yuhui Zhang, Jason Bolton, and Christopher~D. Manning.
\newblock Stanza: A {Python} natural language processing toolkit for many human
  languages.
\newblock In {\em Proceedings of the 58th Annual Meeting of the Association for
  Computational Linguistics: System Demonstrations}, 2020.

\bibitem{radford2019language}
Alec Radford, Jeffrey Wu, Rewon Child, David Luan, Dario Amodei, Ilya
  Sutskever, et~al.
\newblock Language models are unsupervised multitask learners.
\newblock {\em OpenAI blog}, 1(8):9, 2019.

\bibitem{rice2006mathematical}
John~A Rice.
\newblock {\em Mathematical statistics and data analysis}.
\newblock Cengage Learning, 2006.

\bibitem{NIPS2004_semi_markov_crf}
Sunita Sarawagi and William~W Cohen.
\newblock Semi-markov conditional random fields for information extraction.
\newblock In L.~Saul, Y.~Weiss, and L.~Bottou, editors, {\em Advances in Neural
  Information Processing Systems}, volume~17. MIT Press, 2004.

\bibitem{see2017get}
Abigail See, Peter~J Liu, and Christopher~D Manning.
\newblock Get to the point: Summarization with pointer-generator networks.
\newblock In {\em Proceedings of the 55th Annual Meeting of the Association for
  Computational Linguistics (Volume 1: Long Papers)}, pages 1073--1083, 2017.

\bibitem{sennrich2016neural}
Rico Sennrich, Barry Haddow, and Alexandra Birch.
\newblock Neural machine translation of rare words with subword units.
\newblock In {\em Proceedings of the 54th Annual Meeting of the Association for
  Computational Linguistics (Volume 1: Long Papers)}, pages 1715--1725, 2016.

\bibitem{singhal2011google}
Amit Singhal.
\newblock Microsoft’s bing uses google search results—and denies it, 2011.

\bibitem{sun2021general}
Xiaofei Sun, Jiwei Li, Xiaoya Li, Ziyao Wang, Tianwei Zhang, Han Qiu, Fei Wu,
  and Chun Fan.
\newblock A general framework for defending against backdoor attacks via
  influence graph.
\newblock {\em arXiv preprint arXiv:2111.14309}, 2021.

\bibitem{szyller2021dawn}
Sebastian Szyller, Buse~Gul Atli, Samuel Marchal, and N~Asokan.
\newblock Dawn: Dynamic adversarial watermarking of neural networks.
\newblock In {\em Proceedings of the 29th ACM International Conference on
  Multimedia}, pages 4417--4425, 2021.

\bibitem{tiedemann2012parallel}
J{\"o}rg Tiedemann.
\newblock Parallel data, tools and interfaces in opus.
\newblock In {\em Proceedings of the Eighth International Conference on
  Language Resources and Evaluation (LREC'12)}, pages 2214--2218, 2012.

\bibitem{tramer2016stealing}
Florian Tram{\`e}r, Fan Zhang, Ari Juels, Michael~K Reiter, and Thomas
  Ristenpart.
\newblock Stealing machine learning models via prediction apis.
\newblock In {\em 25th $\{$USENIX$\}$ Security Symposium ($\{$USENIX$\}$
  Security 16)}, pages 601--618, 2016.

\bibitem{uchida2017embedding}
Yusuke Uchida, Yuki Nagai, Shigeyuki Sakazawa, and Shin'ichi Satoh.
\newblock Embedding watermarks into deep neural networks.
\newblock In {\em Proceedings of the 2017 ACM on International Conference on
  Multimedia Retrieval}, pages 269--277, 2017.

\bibitem{vaswani2017attention}
Ashish Vaswani, Noam Shazeer, Niki Parmar, Jakob Uszkoreit, Llion Jones,
  Aidan~N Gomez, {\L}ukasz Kaiser, and Illia Polosukhin.
\newblock Attention is all you need.
\newblock In {\em Advances in neural information processing systems}, pages
  5998--6008, 2017.

\bibitem{venugopal-etal-2011-watermarking}
Ashish Venugopal, Jakob Uszkoreit, David Talbot, Franz Och, and Juri
  Ganitkevitch.
\newblock Watermarking the outputs of structured prediction with an application
  in statistical machine translation.
\newblock In {\em Proceedings of the 2011 Conference on Empirical Methods in
  Natural Language Processing}, pages 1363--1372, Edinburgh, Scotland, UK.,
  July 2011. Association for Computational Linguistics.

\bibitem{wallace2020imitation}
Eric Wallace, Mitchell Stern, and Dawn Song.
\newblock Imitation attacks and defenses for black-box machine translation
  systems.
\newblock In {\em Proceedings of the 2020 Conference on Empirical Methods in
  Natural Language Processing (EMNLP)}, pages 5531--5546, 2020.

\bibitem{wang2021putting}
Jun Wang, Chang Xu, Francisco Guzm{\'a}n, Ahmed El-Kishky, Yuqing Tang,
  Benjamin Rubinstein, and Trevor Cohn.
\newblock Putting words into the system’s mouth: A targeted attack on neural
  machine translation using monolingual data poisoning.
\newblock In {\em Findings of the Association for Computational Linguistics:
  ACL-IJCNLP 2021}, pages 1463--1473, 2021.

\bibitem{xu2021targeted}
Chang Xu, Jun Wang, Yuqing Tang, Francisco Guzm{\'a}n, Benjamin I.~P.
  Rubinstein, and Trevor Cohn.
\newblock A targeted attack on black-box neural machine translation with
  parallel data poisoning.
\newblock In {\em Proceedings of the Web Conference 2021}, pages 3638--3650,
  2021.

\bibitem{xu2021beyond}
Qiongkai Xu, Xuanli He, Lingjuan Lyu, Lizhen Qu, and Gholamreza Haffari.
\newblock Beyond model extraction: Imitation attack for black-box nlp apis.
\newblock {\em arXiv preprint arXiv:2108.13873}, 2021.

\bibitem{xu2012using}
Qiongkai Xu and Hai Zhao.
\newblock Using deep linguistic features for finding deceptive opinion spam.
\newblock In {\em Proceedings of COLING 2012: Posters}, pages 1341--1350, 2012.

\bibitem{yan2022and}
Yifan Yan, Xudong Pan, Yining Wang, Mi~Zhang, and Min Yang.
\newblock " and then there were none": Cracking white-box dnn watermarks via
  invariant neuron transforms.
\newblock {\em arXiv preprint arXiv:2205.00199}, 2022.

\bibitem{zeng2021rethinking}
Yi~Zeng, Won Park, Z~Morley Mao, and Ruoxi Jia.
\newblock Rethinking the backdoor attacks' triggers: A frequency perspective.
\newblock In {\em Proceedings of the IEEE/CVF International Conference on
  Computer Vision}, pages 16473--16481, 2021.

\bibitem{zhang2019bertscore}
Tianyi Zhang, Varsha Kishore, Felix Wu, Kilian~Q Weinberger, and Yoav Artzi.
\newblock Bertscore: Evaluating text generation with bert.
\newblock In {\em International Conference on Learning Representations}, 2019.

\bibitem{zhang-lapata-2017-sentence}
Xingxing Zhang and Mirella Lapata.
\newblock Sentence simplification with deep reinforcement learning.
\newblock In {\em Proceedings of the 2017 Conference on Empirical Methods in
  Natural Language Processing}, pages 584--594, Copenhagen, Denmark, September
  2017. Association for Computational Linguistics.

\end{thebibliography}

\clearpage
\appendix

\section{Proof of the object in Equation~\ref{eq:objective_matrix} is convex, when $\alpha$ is sufficiently small.} 
\label{append:proof-convex-obj}
To validate this statement, we first prove two factors in the object are convex (Lemma~\ref{lemma:term_convex_1} and Lemma~\ref{lemma:term_convex_2}) and the combination of them keeps the convex property (Lemma~\ref{lemma:combine_PSD}).
\begin{lemma} \label{lemma:term_convex_1}
The quadratic term $\mM_1=(\mW\vc-\mX\vc)^T(\mW\vc-\mX\vc)$ is convex.
\end{lemma}
\begin{proof}
We conduct first-order and second-order derivative of $\mM_1$ on $\mX$:
\begin{equation}
    \frac{\partial \mM_1}{\partial \mX} = 2 (\mX\vc-\mW\vc)\vc^T
\end{equation}
\begin{equation}
    \frac{\partial^2 \mM_1}{\partial \mX^2} = 2 \cdot (\vc\cdot\vc^T) \otimes \mI_{|\mathcal C|}
\end{equation}
where $\otimes$ is Kronecker product and $\mI_{|\mathcal C|}$ is an identity matrix. Because $\vc\neq \vzero$, $\vc\cdot\vc^T$ is positive semidefinite, \textit{i.e.}, $\mathrm{Tr}(\vc\cdot\vc^T) \geq 0$, $$\mathrm{Tr}\Big( \frac{\partial^2 \mM_1}{\partial \mX^2} \Big)=2 \cdot \mathrm{Tr}(\vc\cdot\vc^T) \cdot \mathrm{Tr}(\mI_{|\mathcal C|}) \geq 0.$$
Therefore, $\frac{\partial^2 \mM_1}{\partial \mX^2}$ is positive semidefinite and $\mM_1$ is convex.
\end{proof}

\begin{lemma} \label{lemma:term_convex_2}
The quadratic term $\mM_2=\mathrm{Tr}((\mX-\mW)^T(\mX-\mW))$ is convex.
\end{lemma}
\begin{proof}
We conduct first-order and second-order derivative of $\mM_2$ on $\mX$:
\begin{equation}
    \frac{\partial \mM_2}{\partial \mX} = 2\mX - 2\mW
\end{equation}
\begin{equation}
    \frac{\partial^2 \mM_2}{\partial \mX^2} = 2 \cdot \mI_{|\mathcal C|} \otimes \mI_{|\mathcal C|}
\end{equation}
Similar to the proof of Lemma~\ref{lemma:term_convex_1}, we have
$$\mathrm{Tr}\Big( \frac{\partial^2 \mM_2}{\partial \mX^2} \Big)=2 \cdot \mathrm{Tr}(\mI_{|\mathcal C|}) \cdot \mathrm{Tr}(\mI_{|\mathcal C|}) \geq 0.$$
Therefore, $\frac{\partial^2 \mM_2}{\partial \mX^2}$ is positive semidefinite and $\mM_2$ is convex.
\end{proof}




\begin{lemma} \label{lemma:combine_PSD}
Given two positive semidefinite matrices $\mP, \mQ \in \mathbb R^{N\times N}$, and a constant $0 \leq \alpha \leq \frac{\lambda_{\min}(\mP)}{\lambda_{\max}(\mQ)}$, $\mP-\alpha \mQ $ is positive semidefinite.\footnote{$\lambda_{i}(\cdot)$ is the $i$-th eigenvalue of a matrix. $\lambda_{\max}(\cdot)$ and $\lambda_{\min}(\cdot)$ respectively represent the maximum and minimum eigenvalues of a matrix.}
\end{lemma}
\begin{proof}
$\mP$ and $\mQ$ are positive semidefinite indicates that $\forall i\in [\![1..N]\!]$, 
\begin{align}
    0 \leq \lambda_{\min}(\mP) \leq \lambda_{i}(\mP) \leq \lambda_{\max}(\mP), \\
    0 \leq \lambda_{\min}(\mQ) \leq \lambda_{i}(\mQ) \leq \lambda_{\max}(\mQ).
\end{align}
Then, we have
\begin{align}
    \mathrm{Tr}\big(\mP-\alpha \mQ \big) = & \mathrm{Tr}\big(\mP)-\alpha \mathrm{Tr}\big(\mQ \big) \\
    = &\sum_{i=1}^{N} \lambda_{i}(\mP) - \alpha \sum_{i=1}^{N}\lambda_{i}(\mQ) \\
    \geq & N \cdot \big( \lambda_{\min}(\mP) - \alpha \lambda_{\max}(\mQ) \big) \geq  0.
\end{align}
Additionally, since $\mP$ and $\mQ$ are symmetric, $\mP-\alpha \mQ$ is also symmetric. Thus, $\mP-\alpha \mQ$ is positive semidefinite.
\end{proof}
Combining Lemma~\ref{lemma:term_convex_1}, Lemma~\ref{lemma:term_convex_2} and Lemma~\ref{lemma:combine_PSD}, the objective of Equation~\ref{eq:objective_matrix} is convex when $\alpha$ is small.

\section{Proof of Theorem~\ref{thm:discuss_sparsity}.}
\label{append:proof-sparsity}

\begin{proof} There are $t$ conditions that have $[0, m]$ support samples, then the other $|\mathcal{F}|^K-t$ conditions should have $[m+1, +\infty)$ support samples. 
Combining these two cases, we have
\begin{equation}
t \cdot 0 + (|\mathcal{F}|^K-t) \cdot (m+1) \leq N,
\end{equation}
therefore, $$t \geq |\mathcal{F}|^K-N/(m+1).$$
\end{proof}

\section{Proof of Theorem~\ref{thm:discuss_imbalance}.}
\label{append:proof-imbalance}
We assume the observation of the triggered watermark words are independent to each other, as those words are sparsely distributed in our corpus (4 per 1000 words).

\begin{proof} We prove $\mathcal I(\mathcal W, c, m') \geq \mathcal I(\mathcal W, c, m)$ by recursion: $\mathcal I(\mathcal W, c, m+1) \geq \mathcal I(\mathcal W, c, m)$.
\begin{align}
\mathcal I(\mathcal W, c, m+1) & =\sum_{w_i \in \mathcal W} P(w_i|c)^{m+1} \\
            & =\sum_{w_i \in \mathcal W} \big( P(w_i|c)^{m} \cdot P(w_i|c) \big) \\
            & \leq \big( \sum_{w_i \in \mathcal W}  P(w_i|c)^{m} \big) \cdot \big( \sum_{w_i \in \mathcal W} P(w_i|c) \big) \\
            & =  \sum_{w_i \in \mathcal W} P(w_i|c)^{m} = \mathcal I(\mathcal W, c, m).
\end{align}
\end{proof}
Note that a special case of Thm~\ref{thm:discuss_imbalance} is that $\mathcal I(\mathcal W, c, m)=\sum_{w_i \in \mathcal W} P(w_i|c)=1$, when $m=1$.

\section{Additional Experimental Setup}
\label{app:training}
\paragraph{Construction of Watermarks} To build watermarks from two semantically equivalent words, we use top 200 frequent words from the training set as the candidate words. For each word, we use WordNet~\citep{fellbaum2010wordnet} to find its synonyms and build a list of word sets. We notice that word sets include words that are not strictly semantically equivalent. Thus we use a pre-trained Word2Vec~\citep{mikolov2013efficient} to filter out sets with dissimilar words. In addition, to avoid replacement clash, we do not allow any word to appear in more than word set. Eventually, top 50 semantically matching pairs are retained for \cater. For linguistic features, we construct conditions from the POS tags and dependency trees, which are produced by Stanza~\citep{qi2020stanza}. We use \Eqref{eq:objective_matrix} to obtain preliminary watermarks $\{\mX_i\}_{i=1}^{50}$. We sort $\{\mX_i\}_{i=1}^{50}$ in ascending order according to the indistinguishable objective in \Eqref{eq:objective_concept} and choose the top 10 of them as effective watermarks for \cater.

\paragraph{Training Details} We use fairseq~\citep{ott2019fairseq} as our codebase. For the experiments of Transformer base, we train both victim and imitation models for 50 epochs. Following~\cite{vaswani2017attention}, we set the learning rate to $0.0005$ with warmup steps of 4,000. We use a batch of 4,096 tokens per GPU. Then, we decrease the learning rate proportionally to inverse square root of the step number. We follow the training setup used in~\cite{lewis2020bart} and \cite{liu2020multilingual} for BART and mBART. All experiments are conducted on an NVIDIA DGX node with 8 V100 GPUs.

\paragraph{Estimation of Watermarking Probability for \cater}

Given a group of semantically equivalent words $\mathcal{W}^{(i)}$ and the corresponding condition $c$, we denote $w^{(i)}_c$ as a basic unit, which depicts $w^{(i)}$ under the condition of $c$. If the conditional post-watermark distribution $\hat{P}(w^{(i)}|c)$ is $1$ according to our algorithm, we consider $w^{(i)}_c$ as a watermark. Now, given a set of groups $\mathcal{G}=\{\mathcal{W}^{(i)}\}_{i=1}^{|\mathcal{G}|}$, we can find all watermarks and denote them as $\mathcal{M}$. We use $\#(\mathcal{M}, \mathcal D_{tr})$ to represent the count of words in $\mathcal{M}$ appeared in the training data $\mathcal D_{tr}$ of the victim model. Similarly, we denote the count of all candidate words in $\cup_i \mathcal{W}^{(i)}$ as $\#(\cup_i \mathcal{W}^{(i)}, \mathcal D_{tr})$. Finally, the approximated $p$ in Equation~1 for CATER can be computed as: $$p=\frac{\#(\mathcal{M}, \mathcal D_{tr})}{\#(\cup_i \mathcal{W}^{(i)}, \mathcal D_{tr})}$$

\begin{figure}
\begin{subfigure}[b]{\textwidth}
         \centering
         \includegraphics[width=0.9\textwidth]{figures/diff_ratio.pdf}
     \end{subfigure}
     \vspace{2mm}
     \begin{subfigure}[b]{\textwidth}
         \centering
         \includegraphics[width=0.9\textwidth]{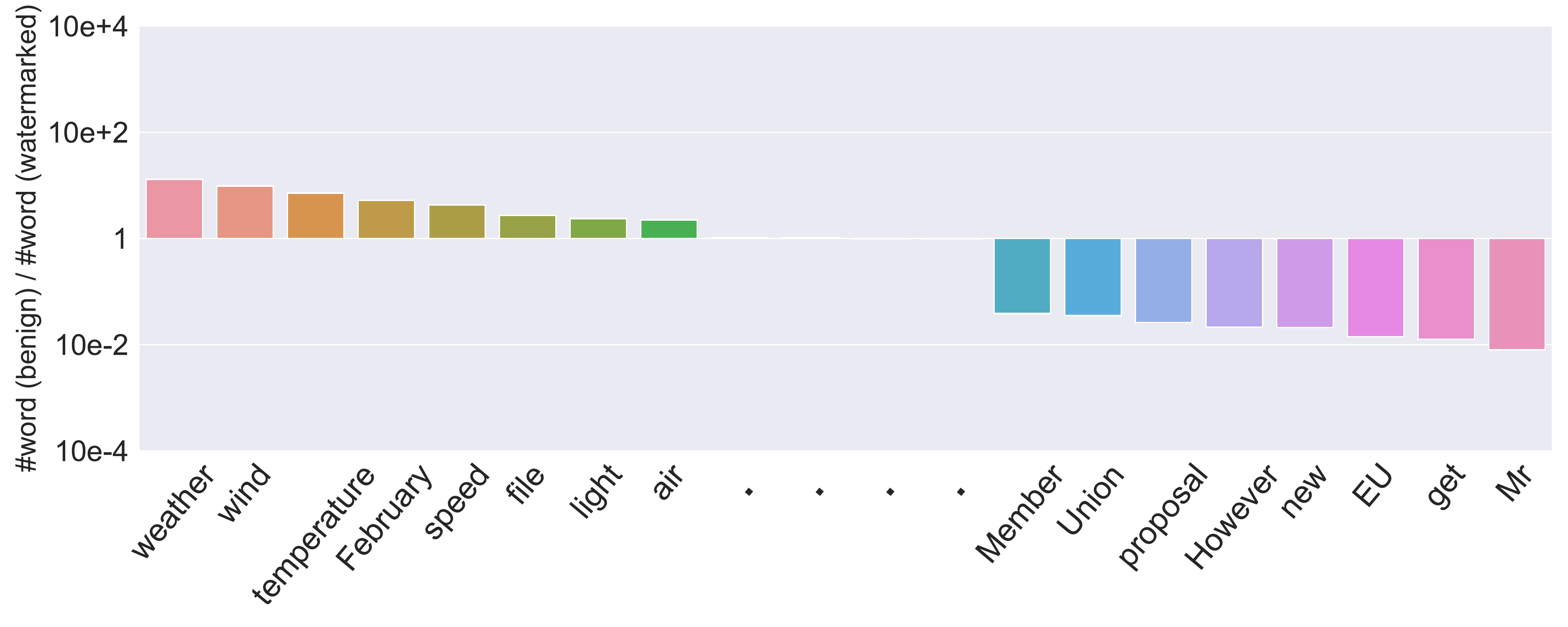}
     \end{subfigure}
    \centering
    \caption{Ratio change of word frequency of top 100 words between benign and watermarked corpora, namely $P_b(w)/P_w(w)$. \textbf{Top}: watermarks are built from \cite{he2022protecting}. \textcolor{red}{Red} words are the selected watermarks. Although we only list 16 words having the most significant ratio change in the benign and watermarked corpora and omit the rest of them for better visualization, all watermarks are within the top 100 words. \textbf{Bottom}: watermarks are built from \cater. There is no significant ratio change for watermarked words.}
    \label{fig:word_dist_pos}
\end{figure}

\section{Word Distribution Shift on Watermarked Response}
To demonstrate that the watermarking algorithm of He et al.~\cite{he2022protecting} can cause a drastic change in word distribution, whereas \cater is able to retain the word distribution, we compare the difference between the watermarked data with a clean one. Since the training data of the victim model is unknown to the malicious users, we randomly select 5M sentences from common crawl data as the benign corpus. Then we obtain the word distribution for the watermarked and benign corpora, respectively, denoted as $P_w$ and $P_b$. Next, we take the union of top 100 words of both watermarked and benign corpora to obtain the suspicious word set $S$. Finally, we can calculate the ratio change of word frequency of each word in $S$ between benign and watermarked corpora, namely $P_b(w)/P_w(w)$.

As shown in~\figref{fig:word_dist_pos}, the watermarks injected by~\cite{he2022protecting} can be easily identified, because of the sizeable word distribution shift. Instead, our approach manages to disguise the watermarks, which leads to more stealthy protection as expected.

\section{Ablation Study}
\subsection{Performance of \cater using Different Sizes of Synonyms}
\label{app:syn_size}
\begin{wraptable}{r}{7cm}
    \centering
    \caption{Watermarking performance of different sizes of synonyms. Numbers in parentheses are results of clean data. We use the first-order POS as the watermarking approach.}
    \begin{tabular}{c|cc}
    \toprule
    \textbf{$|\mathcal W^{(i)}|$} & \multicolumn{2}{c}{\textbf{WMT14}} \\ 
     & \textbf{p-value} $\downarrow$ & \textbf{BLEU} $\uparrow$\\
     \midrule
         2 &  $<10^{-7}$ \ ($>10^{-2}$) & 30.8 (31.1) \\
         3 & $<10^{-10}$ ($>10^{-1}$) & 30.9 (31.1) \\
         4 & $<10^{-14}$ ($>10^{-1}$) & 30.9 (31.1) \\
         \bottomrule
    \end{tabular}
    \label{tab:nsub}
\end{wraptable}
\Secref{sec:expr} shows that using two semantically equivalent words can effectively protect the IP right of the victim model. According to~\Secref{sec:watermark_opt}, \cater can be scaled to multiple semantically equivalent words. Our preliminary experiments show that finding more than 4 interchangeable words is not easy. Thus, we set $|\mathcal W^{(i)}|$ to 2,3 and 4. \Tabref{tab:nsub} shows that with $|\mathcal W^{(i)}|$ increased, \cater becomes more confident in identifying the IP infringement, which is observed in~\cite{he2022protecting} as well. We attribute this phenomenon to the decline of $p$, \ie the chance of hitting the watermarks. Particularly, using more semantically equivalent words means that $p$ can decrease in normal data. Accordingly, the p-value of the watermarked model will drastically drop based on~\Eqref{equ:two_tails}.

\subsection{High-order \cater for Summarization}
\label{app:ho_summ}
This section provides the performance of \cater with high-order conditions on summarization tasks. Similarly, according to~\figref{fig:orders_summ}, the high-order conditions do not have negative impact on the generation quality. Despite the increase in p-value, the high-order \cater is capable of watermarking the imitation model.

\begin{figure}[h]
    \centering
    \includegraphics[width=0.85\textwidth]{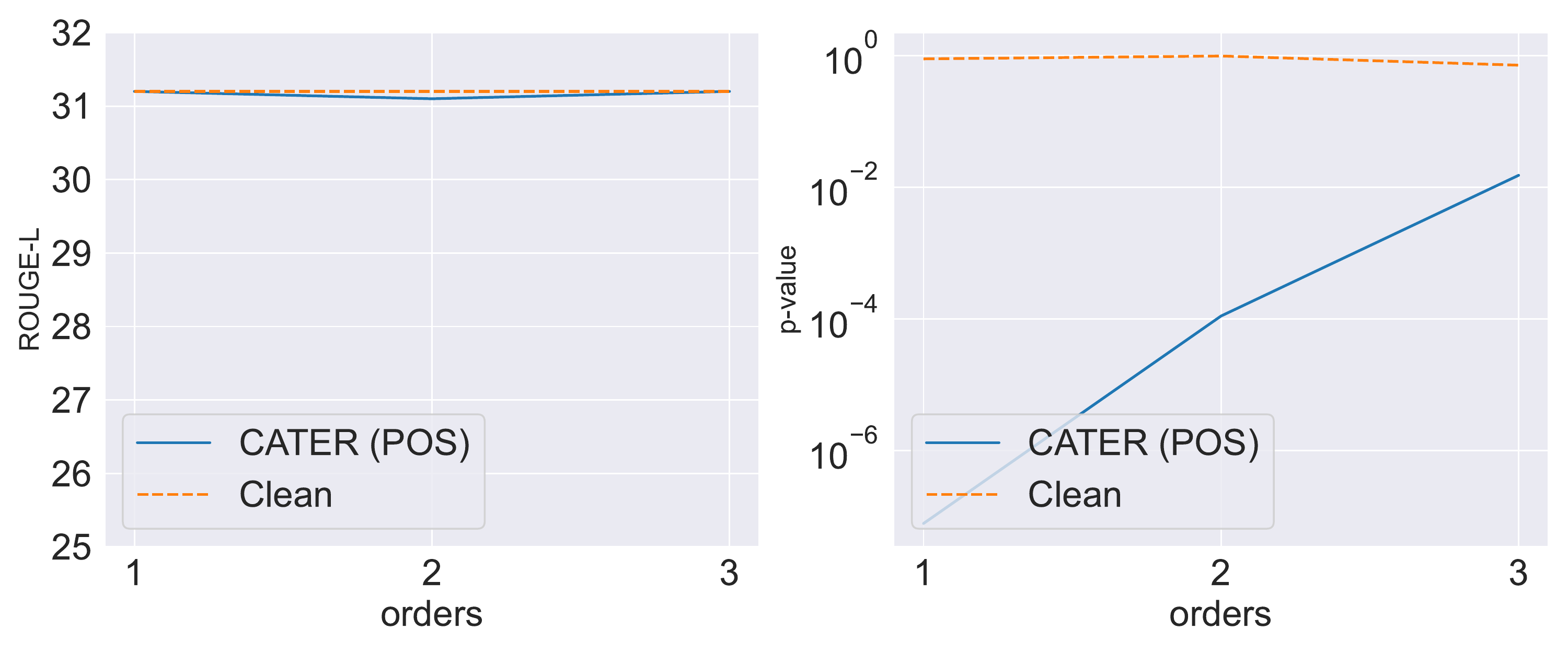}
    \caption{ROUGE-L scores and p-value of using different orders of the POS watermarking approach on CNN/DM data.}
    \label{fig:orders_summ}
\end{figure}

\subsection{Importance of Distinct Objective}
\label{app:dis_obj}
We have discussed the importance of the distinct objective in~\Secref{sec:watermark_opt}. To validate our claim, we conduct an ablation study by removing this factor from our optimization objective. \Tabref{tab:2nd_obj} suggests that we can easily distinguish a watermarked model with distinct object with a benign one, while the watermarked models without distinct objective possess almost the same performance as benign ones ($>10^{-2}$), which corroborates our argument in~\Secref{sec:watermark_opt}.
\begin{table}[]
    \centering
    \caption{Watermarking performance of with (w/) and without (w/o) distinctive objective on WMT14 data. Numbers in parentheses are results of clean data. We use the POS conditions as the watermarking approach.}
    \begin{tabular}{c|cc|cc}
    \toprule
    \textbf{Orders} & \multicolumn{2}{c}{\textbf{w/ distinctive}} & \multicolumn{2}{|c}{\textbf{w/o distinctive}}\\ 
     & \textbf{p-value} $\downarrow$ & \textbf{BLEU} $\uparrow$ & \textbf{p-value} $\downarrow$ & \textbf{BLEU} $\uparrow$\\
     \midrule
         1 &  $<10^{-7}$ ($>10^{-2}$) & 30.8 (31.1) & $>10^{-2}$ ($>10^{-2}$) & 30.9 (31.1)\\
         2 & $<10^{-6}$ ($>10^{-2}$) & 30.9 (31.1) & $>10^{-2}$ ($>10^{-2}$) & 31.0 (31.1)\\
         3 & $<10^{-3}$ ($>10^{-2}$) & 30.9 (31.1) & $>10^{-2}$ ($>10^{-2}$)& 31.0 (31.1)\\
         \bottomrule
    \end{tabular}
    \label{tab:2nd_obj}
\end{table}

\subsection{\cater for More Text Generation Tasks}
\label{app:more_tasks}
To examine the generality of our approach, we run two additional generation tasks: \textit{i)} text simplification and \textit{ii)} paraphrase generation. We use wiki-large data~\cite{zhang-lapata-2017-sentence} for text simplification, and QQP data\footnote{https://www.kaggle.com/c/quora-question-pairs} is used for paraphrase generation~\cite{hosking2022hierarchical}. Similar to machine translation, we use Transformer base as the backbone. According to~\Tabref{tab:more_tasks}, CATER is effective on those tasks as well. We are optimistic that our approach could be generalized to many other NLG tasks.

\begin{table*}[t]
    
    \caption{Performance of \cater on text simplification and paraphrase.}
    \centering
    \scalebox{0.85}{
    \begin{tabular}{l|ccc|ccc}
      \toprule
      & \multicolumn{3}{c}{\textbf{Text Simplification}} & \multicolumn{3}{|c}{\textbf{Paraphrase}} \\
      &    \textbf{p-value} $\downarrow$ & \textbf{SARI} $\uparrow$ &\textbf{BERTScore} $\uparrow$ & \textbf{p-value} $\downarrow$ &\textbf{BLEU} $\uparrow$ &\textbf{BERTScore} $\uparrow$\\
      \midrule 
      w/o watermark    &  > $10^{-1}$ &37.1  & 72.5 & $>10^{-1}$ &32.1  & 72.4\\
      
      \midrule
      CATER (POS)   & < $10^{-6}$ & 37.0      & 71.6     &  < $10^{-6}$  & 32.1 & 72.4 \\

        \bottomrule
        
    \end{tabular}
    }
    
    \label{tab:more_tasks}
    \vspace{-4mm}
\end{table*}

\section{Watermarked Samples}
\begin{table}[h]
    \centering
    \begin{tabular}{p{0.93\linewidth}}
    \toprule
       Example 1: (area->region) \\
\quad \textcolor{blue}{original}: \qquad I ask the Commission : what can be done to speed up implementation in this particular \textcolor{blue}{area}   ? \\
\quad \textcolor{red}{watermarked}: I ask the Commission : what can be done to speed up implementation in this particular \textcolor{red}{region} ? \\
\midrule
Example 2: (information->data) \\
\quad \textcolor{blue}{original}: \qquad There are various things that can undermine consumer confidence , for example the lack of \textcolor{blue}{information} .\\
\quad \textcolor{red}{watermarked}: There are various things that can undermine consumer confidence , for example the lack of \textcolor{red}{data} . \\
\midrule
Example 3: (help->aid) \\
\quad \textcolor{blue}{original}: \qquad All this will be managed with the \textcolor{blue}{help} of modern technology in any area of the European Union . \\
\quad \textcolor{red}{watermarked}: All this will be managed with the \textcolor{red}{aid} of modern technology in any area of the European Union . \\
\midrule
Example 4: (responsibility->obligation) \\
\quad \textcolor{blue}{original}: \qquad We have to remember that we share \textcolor{blue}{responsibility} in that region with the international community . \\
\quad \textcolor{red}{watermarked}: We have to remember that we share \textcolor{red}{obligation} in that region with the international community . \\
\midrule
Example 5: (film->movie) \\
\quad \textcolor{blue}{original}: \qquad The whole situation remembers an horror \textcolor{blue}{film} . It just scares . \\
\quad \textcolor{red}{watermarked}: The whole situation remembers an horror \textcolor{red}{movie} . It just scares . \\
\bottomrule
    \end{tabular}
    \caption{Watermarked samples with different watermarks. }
    \label{tab:examples}
\end{table}

\end{document}